\documentclass[11pt]{article}

\usepackage[left=2.5cm,right=2.5cm,top=2.5cm,bottom=2.5cm]{geometry}

\usepackage[english]{babel}
\usepackage[applemac]{inputenc}       
\usepackage[T1]{fontenc}

\usepackage{multirow}
\usepackage{tabularx}
\newcolumntype{s}{>{\hsize=.05\hsize \raggedright\arraybackslash}X} 
\newcolumntype{t}{>{\hsize=.15\hsize \raggedright\arraybackslash}X} 
\newcolumntype{b}{>{\hsize=.95\hsize}X}
\usepackage{array}
\usepackage{arydshln}
\usepackage{pdflscape}
\usepackage{afterpage}
\usepackage{lscape}

\usepackage{amsmath}
\usepackage{amssymb}

\usepackage{amsthm}
\usepackage{mathtools}
\usepackage{mathrsfs}

\usepackage{caption}
\usepackage{subcaption}
\usepackage{makecell}
\usepackage{todonotes}

\usepackage{enumerate}

\usepackage{dsfont}

\usepackage{graphicx}

\usepackage{tikz}
\usetikzlibrary{calc}
\usetikzlibrary{arrows}
\usetikzlibrary{positioning}
\usetikzlibrary{fit}

\usepackage{pifont}
\newcommand{\norm}[1]{\left\lVert#1\right\rVert}
\newcommand{\abs}[1]{\left\lvert#1\right\rvert}

\theoremstyle{plain}
\newtheorem{thm}{Theorem}
\newtheorem{cor}[thm]{Corollary}
\newtheorem{lem}[thm]{Lemma}
\newtheorem{prop}[thm]{Proposition}

\theoremstyle{definition}
\newtheorem{defn}[thm]{Definition}
\newtheorem{rem}[thm]{Remark}

\newtheorem{exmp}[thm]{Example}

\newcommand{\einv}{{\epsilon^{-1}}}

\newcommand{\ee}{{\mathrm e}}
\newcommand{\ii}{{\mathrm i}}
\newcommand{\dd}{{\mathrm d}}

\DeclareMathOperator{\tr}{tr}
\DeclareMathOperator{\rk}{rk}
\DeclareMathOperator{\ran}{ran}
\DeclareMathOperator{\sign}{sign}
\let\Re\relax
\let\Im\relax
\DeclareMathOperator{\Re}{Re}
\DeclareMathOperator{\Im}{Im}
\newcommand{\bvec}[1]{\vec{#1}}
\renewcommand{\bvec}[1]{\boldsymbol{#1}}

\newcommand{\He}{H^\sharp}

\DeclarePairedDelimiterX\braket[2]{\langle}{\rangle}{#1 , #2}

\DeclareFontFamily{U}{mathx}{\hyphenchar\font45}
\DeclareFontShape{U}{mathx}{m}{n}{
      <5> <6> <7> <8> <9> <10>
      <10.95> <12> <14.4> <17.28> <20.74> <24.88>
      mathx10
      }{}
\DeclareSymbolFont{mathx}{U}{mathx}{m}{n}
\DeclareFontSubstitution{U}{mathx}{m}{n}
\DeclareMathAccent{\widecheck}{0}{mathx}{"71}
\DeclareMathAccent{\wideparen}{0}{mathx}{"75}

\usepackage{tcolorbox}

\usepackage[pdftex,%
colorlinks=true,linkcolor=blue,citecolor=blue,urlcolor=blue
]
{hyperref}

\numberwithin{equation}{section}

\graphicspath{{Figures/}}

\usepackage{authblk}

\title{Classifying bulk-edge anomalies in the Dirac Hamiltonian}
\author[1]{Hansueli Jud}
\author[2]{Clément Tauber\thanks{tauber@ceremade.dauphine.fr}}
\affil[1]{Lab42,
	Obere Strasse 22B,
	7270 Davos, Switzerland
}
\affil[2]{CEREMADE, CNRS, Université Paris-Dauphine, Université PSL, 75016 PARIS, FRANCE}
\date{\today}

\begin{document}

\maketitle

\begin{abstract}
	We study the Dirac Hamiltonian in dimension two with a mass term and a large momentum regularization, and show that bulk-edge correspondence fails. Despite a well defined bulk topological index --the Chern number--, the number of edge modes depends on the boundary condition. The origin of this anomaly is rooted in the unbounded nature of the spectrum. It is detected with Levinson's theorem from scattering theory and quantified via an anomalous winding number at infinite energy, dubbed ghost charge. First we classify, up to equivalence, all self-adjoint boundary conditions, using Schubert cell decomposition of a Grassmanian. Then, we investigate which ones are anomalous. We expand the scattering amplitude near infinite energy, for which a dominant scale captures the asymptotic winding number. Remarkably, this can be achieved for every self-adjoint boundary condition, leading to an exhaustive anomaly classification. It shows that anomalies are ubiquitous and stable. Boundary condition with a ghost charge of $2$ is also revealed within the process. 
\end{abstract}

\section{Introduction}

Bulk-edge correspondence is a profound result of topological insulators. It has been used as a paradigm for more than 40 years to understand the exotic properties of topological materials. The bulk index, defined for an infinite sample without boundary, predicts the number of states appearing at the edge of a half-infinite sample with a sharp boundary. Such edge states are confined near the boundary and have remarkable properties, like robust unidirectional propagation in two dimensions. Their number is counted by the edge index, which equals the bulk index via the bulk-edge correspondence.

Such a result was established with various techniques from spectral analysis to K-theory and for a very broad class of models, originally in the context of the quantum Hall effect \cite{Halperin82,Hatsugai93} and then extended to various condensed matter systems \cite{AvronSeilerSimon94,BellissardVanElstSchulzBaldes94,SchulzBaldesKellendonkRichter00,GrafElbau02,CombesGerminet05,GrafPorta13,Avilaetal13,EssinGurarie11,ProdanSchulzBaldes16,MathaiThiang16,BourneRennie18,Drouot19,GomiThiang19,Corneanetal21,Corneanetal23, Bal23} and actually to almost any energy-preserving classical wave phenomenon \cite{RaghuHaldane08,DeNittisLein17,Perietal19,DelplaceMarstonVenaille17}. Yet, it has its limitation: Bulk-edge correspondence has been shown to fail in the context of shallow-water waves \cite{TauberDelplaceVenaille19bis,GrafJudTauber21}. There, the system has a well-defined bulk index and yet the number of edge modes depends on the choice of boundary condition.  The origin of the anomaly was identified in the unbounded nature of the operator spectrum. Using scattering theory and a relative version of Levinson's theorem, a winding number captures the anomalous asymptotic properties of the spectrum and compensates for the lack or excess of edge modes in the finite spectrum. Such a winding number was dubbed ``ghost'' charge in \cite{TauberDelplaceVenaille19bis}, with possible physical consequences. 

The purpose of this paper is to show that this anomalous situation is not isolated. First, it is not restricted to water waves. As noticed in \cite{TauberDelplaceVenaille19bis}, we show that it also appears in the celebrated Dirac Hamiltonian, a canonical model to describe graphene near Fermi energy \cite{Castroetal09}. Its massive version is used as an effective theory for various topological materials, like the Haldane or Kane-Mele model  \cite{Haldane88,KaneMele02}. However, because its bulk index is ill-defined, we consider instead its regularized version which compactifies the problem at large momenta. The regularized and massive Dirac Hamiltonian appears for example in \cite{Volovik88} in the context of superfluid Helium. Second, and more importantly, we do not study a specific anomalous boundary condition as in \cite{GrafJudTauber21} but classify them all. 

The first result of this paper is to classify all self-adjoint boundary conditions for the massive and regularized Dirac Hamiltonian. We identify them with stable subspaces of $2 \times 4$ matrices of maximal rank. Up to equivalence, we get seven classes based on the Schubert cell decomposition of the Grassmannian $\mathrm{Gr}(2,4)$. The main result of this paper is then to compute the asymptotic winding number, denoted $w_\infty \in \mathbb Z$ below,  for any boundary condition among theses classes. Remarkably, the classification can be done exhaustively. One consequence is that anomalies are actually ubiquitous and very stable: this is not a fine-tuning effect. Nevertheless, many boundary conditions (e.g. Dirichlet) are also anomaly-free and stable as well.  

To do so, we develop a formalism to extract precisely the dominant scale of the scattering amplitude in the asymptotic part of the spectrum, in order to compute the winding number  $w_\infty$. Then we manage to apply this strategy to any boundary condition, class by class. In some cases the expressions contain up to 60 terms, but using a symbolic computation software --Mathematica-- our approach reduces them to small and tractable complex curves. Surprisingly, among the many parameters in each class, only one or two drive the anomaly, so that $w_\infty$ can be always computed exactly and moreover the whole classification can be summarized in a rather compact way. Finally, the classification  also reveals a family of anomalous boundary conditions where the ghost charge is $2$ instead of $1$, which was unexpected and never noticed before. 

The general message of this paper is that one should always consider with extra care apparent topological edge modes of unbounded operators, as they might be anomalous. However, there are several way to circumvent anomalous boundary conditions. Considering a smooth interface or domain wall instead of a sharp boundary usually removes the anomaly and restores the bulk-edge correspondence \cite{TauberDelplaceVenaille19,RossiTarantola24}. Moreover, there are situations where an interface index can be properly defined even though there is no bulk index \cite{DelplaceMarstonVenaille17,Delplace22, Bal22, QuinnBal22}.  Furthermore it is also possible to avoid the asymptotic part of the spectrum even with a sharp boundary condition, for example via a spectral flow of boundary conditions \cite{TauberThiang23} or using an external scalar field \cite{OnukiVenailleDelplace23}. Finally, we also mention a related anomalous problem for the magnetic Dirac Hamiltionian in graphene \cite{Treustetal24}, in a slightly different context. 

The rest of the paper is organized as follows: In Section~\ref{sec:main} we describe the model, explain the scattering theory to define $w_\infty$ and present the main results: Theorem~\ref{thm:SABC} and \ref{thm:classification}, associated with Table~\ref{tab:self-adjoint_classes} and \ref{tab:anomalies}, respectively. Section~\ref{sec:proof_sabc} proves the boundary condition classification. Section~\ref{sec:asym_S} develops the main tools to expand the scattering amplitude near infinity, and Section~\ref{sec:anomaly_classification} applies these tools to classify all the anomalies.

\section{Setting and main results\label{sec:main}}

\subsection{Dirac Hamiltonian and its topological index}

The massive and regularized Dirac Hamiltonian is a self-adjoint operator on $L^2(\mathbb R^2)$ with domain $H^1(\mathbb R^2) \oplus H^1(\mathbb R^2)$ given by
\begin{equation}\label{eq:Dirac_Hamiltonian}
	\mathcal{H} = \begin{pmatrix}
			m + \epsilon (\partial_x^2 + \partial_y^2) & \ii \partial_x + \partial_y\\
			\ii \partial_x -\partial_y & -m -\epsilon (\partial_x^2 +\partial_y^2)
	\end{pmatrix}\,.
\end{equation}
In the following we assume $m>0$, $\epsilon>0$ and $\epsilon< 1/2 m$. The mass parameter $m$ opens a gap in the spectrum, see \eqref{eq:bulkbands} below, and the term $\epsilon (\partial_x^2 + \partial_y^2)$ regularizes the problem at high energy. It also has a physical interpretation in the context of superfluid Helium \cite{Volovik88}.
By translation invariance in space, the stationary solutions to the Schrödinger equation $\ii \partial_t \psi = \mathcal H \psi$ are given by normal modes
\begin{equation}\label{eq:normal_modes}
	\psi:=\widehat{\psi}(k_x,k_y,\omega)\ee^{\ii (k_x x+k_yx-\omega t)}\,.
\end{equation}
They have momentum $\bvec{k}=(k_x,k_y)\in\mathbb{R}^2$, energy $\omega\in\mathbb{R}$ and correspond to the eigenvalue problem:
\begin{equation}\label{eq:bulk_eigenvalue_equation}
	H\widehat{\psi}=\omega\widehat{\psi}\,,\quad \widehat{\psi}=\begin{pmatrix}
		\widehat{\psi}_1\\
		\widehat{\psi}_2
	\end{pmatrix}\,,\quad H(\bvec{k})=\begin{pmatrix}
		m -\epsilon\bvec{k}^2 & -k_x +\ii k_y\\
		-k_x -\ii k_y & -m +\epsilon\bvec{k}^2
	\end{pmatrix}\,,
\end{equation}
with $\bvec{k}^2=k_x^2+k_y^2$ and $H(\bvec{k})$ a Hermitian matrix. The system admits two energy bands
\begin{equation}\label{eq:bulkbands}
	\omega_\pm = \pm \sqrt{\bvec{k}^2 + \left(m-\epsilon \bvec{k}^2\right)^2}
\end{equation}
separated by a gap of size $m$. We denote by $P_\pm(\bvec{k})$ the corresponding eigenprojections. Since $\epsilon \neq 0$, these projections are single-valued at $\bvec{k}\to \infty$, which allows to consider the associated fibre bundle over the compactified plane $\mathbb R^2 \cup \{\infty\} \cong S^2$. See \cite{GrafJudTauber21} or Appendix~\ref{app:Chern} for more details. Consequently, each projection $P$ has an associated fibre bundle index --the Chern number-- given by
\begin{equation}
	C(P)=\frac{1}{2\pi\ii}\int_{S^2} dk_xdk_y\tr\left(P\left[\partial_{k_x}P,\partial_{k_y}P\right]\right)\,.
\end{equation}
A short computation shows that 
$$
C_\pm := C(P_\pm)= \pm \dfrac{\sign{m}+\sign{\epsilon}}{2} = \pm 1.
$$
Thus, each energy band has an associated non-trivial bulk topological index.

\subsection{Self-adjoint boundary conditions}

We restrict the spatial domain to $\{(x,y)\,\vert\, y\geq 0\}\subset \mathbb{R}^2$ and study the restriction of $\mathcal H$ to it, named $\mathcal H^\sharp$. This domain preserves translation invariance in $x$-direction, so we focus on normal modes of the form $\psi:=\widetilde{\psi}(y;k_x,\omega)\ee^{\ii (k_x x-\omega t)}$. The Schr\"odinger equation $\ii \partial_t \psi = \mathcal H^\sharp \psi$ is  reduced to the study of a one-parameter family of operators on $L^2(\mathbb R^+)$:
\begin{equation}\label{eq:edge_Hamiltonian}
	\He \widetilde \psi = \omega \widetilde \psi, \qquad
	\He(k_x)=\begin{pmatrix}
		m -\epsilon k_x^2 + \epsilon \partial_y^2 & -k_x + \partial_y\\
		-k_x - \partial_y & -m + \epsilon k_x^2 - \epsilon \partial_y
	\end{pmatrix}\,,
\end{equation}
for $k_x,\omega \in \mathbb R$. 
We consider any local boundary conditions for this problem:
\begin{equation}\label{eq:local_boundary_condition}
	\left.(B_0 + i k_x B_1) \widetilde{\psi} + B_2 \widetilde{\psi}^\prime \right|_{y=0}= 0
\end{equation}
where $B_0,B_1,B_2 \in \mathrm{M}_2(\mathbb C)$ are $k_x$ and $y$-independent and
\begin{equation}\label{eq:states_halfspace}
	\widetilde{\psi} = \begin{pmatrix}
		\psi_1 \\ \psi_2
	\end{pmatrix}, \qquad \widetilde{\psi}^\prime = \begin{pmatrix}
	\partial_y \psi_1 \\ \partial_y \psi_2
\end{pmatrix}.
\end{equation}
For example, Dirichlet's boundary condition corresponds to $B_0=I_2$ and $B_1=B_2=0$. However, not all matrices in \eqref{eq:local_boundary_condition} define a self-adjoint operator associated to \eqref{eq:edge_Hamiltonian}. The first result of this paper is to classify all local self-adjoint boundary conditions. It appears convenient to regroup the $B_i$ matrices in  $2\times4$-matrices $A,A_0,A_1 \in \mathrm{M}_{2,4}(\mathbb C)$:
\begin{equation}
	A := A_0 + \ii k_x A_1, \qquad A_0 := [B_0 | B_2], \qquad A_1 := [B_1 | 0],
\end{equation}
so that \eqref{eq:local_boundary_condition} is equivalent to 
\begin{equation}\label{eq:local_boundary_condition_A}
\left.	A \Psi\right|_{y=0} =0, \qquad \Psi = \begin{pmatrix} \widetilde\psi \\ \widetilde\psi^\prime \end{pmatrix} \in \mathbb C^4.
\end{equation}

\begin{defn}
Two boundary conditions given by $A$ and $\widetilde{A}$ are equivalent if $\widetilde{A} = B A$ for some $B \in  \mathrm{GL}_2(\mathbb C)$.
\end{defn}

Equivalent boundary conditions preserve self-adjointness as well as anomalies, see below. 

\begin{thm}\label{thm:SABC}
For almost every $k_x \in \mathbb R$ the operator $\He(k_x)$ together with boundary condition \eqref{eq:local_boundary_condition_A} is self-adjoint on $L^2(\mathbb R_+)$ if and only if $A \in \mathrm{M}_{2,4}(\mathbb C)$ is a rank-2 matrix and  is equivalent to one of the seven classes $\mathfrak A_{i,j}, \mathfrak B$ and $\mathfrak C$ given in Table~\ref{tab:self-adjoint_classes}.
\end{thm}

\begingroup
\setlength{\tabcolsep}{20pt}
\renewcommand{\arraystretch}{1}
\begin{table}[h!]
	\centering
	\begin{tabular}{|c|c|c|}
		\hline
	Class &	$A_0$ & $A_1$	\\
		\hline
		\hline   
		& & \\
		
	$\mathfrak{A}_{1,2}$ &	$\begin{pmatrix}
			1 & 0 & 0 & 0\\
			0 & 1 & 0 & 0
		\end{pmatrix}$	& $\begin{pmatrix}
			b_{11} & b_{12} & 0 & 0\\
			b_{21} & b_{22} & 0 & 0
		\end{pmatrix}$	 \\ 
	
	 		& & \\
		\hline
	 		& & \\
	$\mathfrak{A}_{1,4}$ &	$\begin{pmatrix}
			1 & 0 & 0 & 0\\
			0 & \alpha & 0 & 1
		\end{pmatrix}$ 	& $\begin{pmatrix}
		b_{11} & 0 & 0 & 0\\
		b_{21} & \ii \beta & 0 & 0
	\end{pmatrix}$	\\

&$\alpha\in\mathbb{R}$ &  $\beta\in\mathbb{R}$ \\ 
	\hline
		 		& & \\
	$\mathfrak{A}_{2,3}$ & $\begin{pmatrix}
			0 & 1 & 0 & 0\\
			\alpha & 0 & 1 & 0
		\end{pmatrix}$	& $\begin{pmatrix}
		0 & b_{12} & 0 & 0\\
		\ii \beta & b_{22} & 0 & 0
	\end{pmatrix}$ \\

& $\alpha\in\mathbb{R}$ & $\beta\in\mathbb{R}$ \\
		\hline
	 		& & \\
	$\mathfrak{A}_{2,4}$ &	$\begin{pmatrix}
			a_{11} & 1 & 0 & 0\\
			a_{21} & 0 &(a_{11}^*)^{-1}& 1
		\end{pmatrix}$	& $\begin{pmatrix}
		b_{11} & b_{11}(a_{11})^{-1} & 0 & 0\\
		 b_{21} & b_{22} & 0 & 0
	\end{pmatrix}$ \\

&  $\scriptstyle a_{11}\in \mathbb C\setminus \{0\},\quad a_{21} = \alpha a_{11 }+\epsilon^{-1}, \quad \alpha \in \mathbb R$ &  $\scriptstyle b_{22}-b_{21}a_{11}^{-1}=\ii \beta, \quad \beta \in \mathbb R$ \\
		\hline
	 		& & \\
	$\mathfrak{A}_{3,4}$ &	$\begin{pmatrix}
			\alpha_{1} & a_{12} & 1 & 0\\
			\einv -a_{12}^* & \alpha_{2} & 0 & 1
		\end{pmatrix}$	& $\begin{pmatrix}
		\ii \beta_{1} & b_{12} & 0 & 0\\
		b_{12}^* & \ii \beta_{2} & 0 & 0
	\end{pmatrix}$ \\

& $\alpha_{1}, \alpha_{2}\in\mathbb{R}$ & $\beta_{1},\beta_{2} \in \mathbb R$ \\
		\hline
		\hline
			 		& & \\
	$\mathfrak B$	&$\begin{pmatrix}
			a_1 & a_2 & \ii\alpha & - \ii  \mu^*\alpha\\
			\mu a_1 & \mu a_2 & \ii  \mu  \alpha & - \ii |\mu|^2  \alpha
		\end{pmatrix}$	& $\begin{pmatrix}
			1 & 0 & 0 & 0\\
			0 & 1 & 0 & 0
		\end{pmatrix}$	 \\

	& $\scriptstyle \alpha\in\mathbb{R}$, $\scriptstyle \alpha \big(\alpha \Im(\mu) - \epsilon \Re(a_1-a_2\mu)\big)=0$&\\
	
%
%
		\hline
		\hline
			 		& & \\
	$\mathfrak C$ &	$\begin{pmatrix}
			a_1 & a_2 & 0 & a_4\\
			\mu a_1 & \mu a_2 & 0 & \mu a_4
		\end{pmatrix}$	& $\begin{pmatrix}
			1 & 0 & 0 & 0\\
			0 & 0 & 0 & 0
		\end{pmatrix}$  \\

	& $\scriptstyle (a_2,a_4)\neq 0,\, \mu \in \mathbb C\setminus\{0\},\, \Im\left(a_2a_4^*\right)=0$ & \\
		\hline
	\end{tabular}
	\caption{Classification of local self-adjoint boundary conditions from Theorem~\ref{thm:SABC}, up to $\mathrm{GL}_2$-invariance. All unconstrained parameters are arbitrary complex numbers: $a_{ij}, b_{ij}, a_i \in \mathbb C$.\label{tab:self-adjoint_classes}}
\end{table}
\endgroup

The proof of this theorem can be found in Section~\ref{sec:proof_sabc}. We relate self-adjoint extensions of $H^\sharp$ with stable subspaces generated by $\ker A$, then use the  $\mathrm{GL}_2$-invariance to identify $A$ with elements of the Grassmanian $\mathrm{Gr}(4,2)$. The subclasses $\mathfrak{A}_{i,j}$ correspond to a Schubert cell decomposition which is compatible with the self-adjoint constraint. 

\begin{rem}
$A_0$ and $A_1$ are $k_x$-independent, so each couple $(A_0,A_1)$ provides a  self-adjoint operator $\mathcal H^\sharp$ on $L^2(\mathbb R\times \mathbb R^+)$. It may happen exceptionally that $A=A_0+\ii k_{x,0}A_1$ is not of rank $2$ for some $k_{x,0}\in \mathbb R$ (e.g. $k_{x,0}=0$ in class $\mathfrak{B}$) so that that $H^{\sharp}(k_{x,0})$ is not self-adjoint. However this would happen only for a finite number of $k_x$ so that  $\mathcal H^\sharp$ remains globally self-adjoint.
\end{rem}

\subsection{Edge modes} For each boundary condition, equation \eqref{eq:edge_Hamiltonian} is a system of two coupled ODEs in $y$ which is exactly solvable for every $k_x,\, \omega$.  We consider the spectrum of $\He$, namely solutions to \eqref{eq:edge_Hamiltonian} which are bounded in $y$, and plot it in the $(k_x,\omega)$-plane. We distinguish two kinds of solutions. For $|\omega|\geq \omega_+(k_x,0)$, see \eqref{eq:bulkbands}, the solutions are oscillatory modes delocalized in the upper half-plane, they correspond to the essential spectrum of $\mathcal H$. We call them bulk modes. On the other hand, for   $|\omega|\leq \omega_+(k_x,0)$, there may exist  solutions which decay exponentially when $y \to \infty$. We call them edge modes. As $k_x$ varies, edge modes draw continuous curves that we count as follows.
\begin{defn}
	The number $n_\mathrm{b}$ of edge modes below a bulk band is the signed number of
	edge mode branches emerging $(+)$ or disappearing $(-)$ at the lower band limit, as $k_x \in \mathbb R$ increases.
	The number $n_\mathrm{a}$ of edge modes above a band is counted likewise up to a global sign change.
\end{defn}

\newpage
\begin{exmp}\label{example}
To illustrate the diversity of situations we plot three examples of the spectrum of $\He$ in Figure~\ref{fig:edgemodes} for the following boundary condition:
\begin{itemize}
	\item Dirichlet boundary condition: class $\mathfrak{A}_{12}$ with $A_0= [I_2 | 0 ]$ and $A_1 = [0| 0]$. This corresponds to $\psi_1=\psi_2=0$ at $y=0$.
	\item Condition a: class $\mathfrak A_{1,4}$ with $A_0 = \begin{pmatrix}
		1 &0&0&0\\0 & 0 & 0 & 1
	\end{pmatrix}$ and $A_1 = \begin{pmatrix}
	0 &0&0&0\\0 & -\ii & 0 &0
\end{pmatrix}$ which corresponds to $\psi_1=0$ and $\psi_2^\prime + k_x \psi_2  =0$ at $y=0$.
\item Condition b: class $\mathfrak A_{3,4}$ with 
$A_0=\begin{pmatrix}
	1 & 1 & 1 & 0 \\
	9 & 1 & 0 & 1 \\
\end{pmatrix}
$ and $
A_1=
\begin{pmatrix}
	4 \ii & -\ii & 0 & 0 \\
	\ii & -4 \ii & 0 & 0 \\
\end{pmatrix}$.
\end{itemize}
\begin{figure}[h!]
	\centering
	\includegraphics[width=0.3\textwidth]{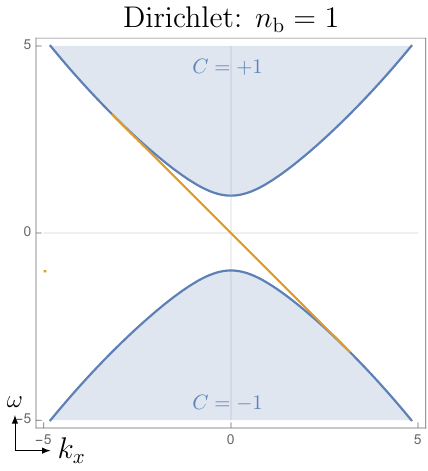}\hspace{0.03\textwidth}
	\includegraphics[width=0.3\textwidth]{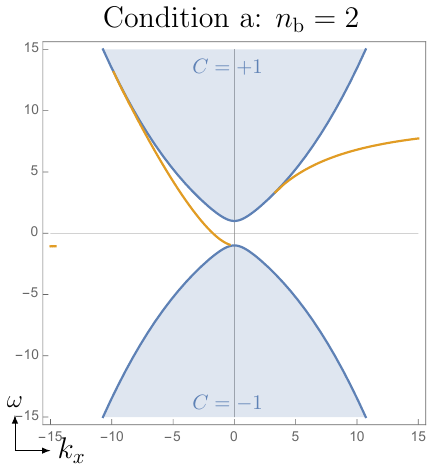}\hspace{0.03\textwidth}
		\includegraphics[width=0.3\textwidth]{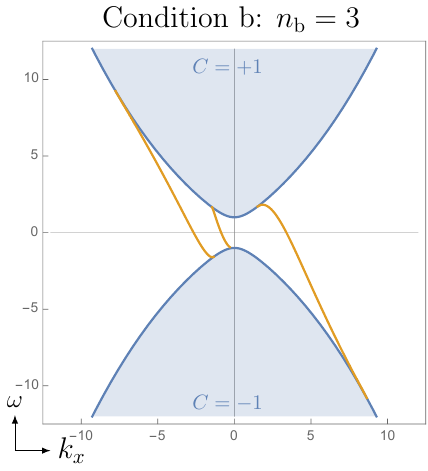}
\caption{Edge spectrum of $\He$ for $m=1$, $\epsilon=0.1$ and various boundary conditions from Example~\ref{example}. Blue regions correspond to bulk modes and yellow curves to edge modes. Their number below the upper band is respectively $n_\mathrm{b} = 1,2$ and $3$.\label{fig:edgemodes}}
\end{figure}
We focus on the edge modes below the upper band. For Dirichlet boundary condition, we get $n_\mathrm{b}=1=C_+$, in agreement with the bulk-edge correspondence. However, such a correspondence is violated for boundary conditions (a) and (b), for which we get $n_\mathrm{b}=2$ and $3$, respectively. Changing $-\ii$ to $\ii$ in $A_1$ of Condition (a) leads to a spectrum with no edge modes, so that $n_\mathrm{b}=0$. As we shall see, for this model all the possible values are $n_\mathrm{b} = -1,0,1,2,3$, see Remark~\ref{rem:mainthm} below.
\end{exmp}


Similarly to \cite{GrafJudTauber21}, a boundary condition with $C_+\neq n_\mathrm{b}$ is said to be anomalous. Notice that the mismatch cannot be cured by looking at higher --but finite-- values of $k_x$ and $\omega$, possibly revealing extra edge modes that are not visible in Figure~\ref{fig:edgemodes}. The anomaly is actually rooted in a singularity at $k_x \to \infty$, and can be precisely detected for any boundary condition via the scattering amplitude formalism.

\subsection{Scattering amplitude}

For the rest of the paper we focus on the upper part of the spectrum $\omega>0$. A similar analysis could be performed on the lower part\footnote{In general, a boundary condition does not preserve the symmetry of the bulk, so the nature --anomalous or not-- of the upper band may be different than the lower one. However the analysis of the lower band would be very similar.}. Let $k_x\in \mathbb R$ and fix $\omega > \omega_+(k_x,0) = \sqrt{k_x^2 + (m-\epsilon k_x^2)^2}$. The bulk band equation
\begin{equation}\label{eq:def_kappa}
	k_x^2 + k_y^2 + (m-\epsilon(k_x^2 + k_y^2))^2 = \omega^2
\end{equation}
has four solutions in $k_y$, two are opposite and real and two are opposite and purely  imaginary. Let $\kappa=\kappa(k_x,\omega)$ be the real positive one. Then the other solutions are: $-\kappa<0$ and 
\begin{equation}\label{eq:kappa_ev_div}
	\kappa_{\mathrm{ev}/\mathrm{div}} = \pm \ii \sqrt{\kappa^2 + 2 k_x^2 + \frac{1-2\epsilon m}{\epsilon^2}} \in \ii\mathbb{R}_+\,.
\end{equation}
Reciprocally, for fixed $k_x\in \mathbb R$ and $\kappa>0$, $k_y= \kappa,-\kappa,\kappa_\mathrm{ev}$ and $\kappa_\mathrm{div}$ are all solution to \eqref{eq:def_kappa} for the same $\omega = \omega_+(k_x,\kappa)$. Let $\widehat \psi_\mathrm{in}$, $\widehat \psi_\mathrm{out}$ and $\widehat \psi_\mathrm{ev}$ be bulk eigenstates, namely solutions of $H\widehat \psi = \omega \widehat\psi$, see \eqref{eq:bulk_eigenvalue_equation}, with momentum $k_x$ and $k_y=-\kappa, \kappa, \kappa_{\mathrm{ev}}$, respectively. The associated normal modes are not solutions of the half-space problem, but a linear combination can be. We define the scattering state as
\begin{equation}\label{eq:def_scattering_state}
\psi_\mathrm{s} := \widetilde \psi_\mathrm{s} \ee^{\ii (k_x x- \omega t)}, \qquad  \widetilde \psi_\mathrm{s}:= \widehat \psi_\mathrm{in} \ee^{-\ii \kappa y}  +  S \widehat \psi_\mathrm{out} \ee^{\ii \kappa y}  + T \widehat \psi_\mathrm{ev} \ee^{\ii \kappa_\mathrm{ev}y}.
\end{equation}
Such a state has well-defined momentum $k_x$ and energy $\omega = \omega_+(k_x,\kappa)$. Each term is respectively interpreted as an incoming, outgoing and evanescent state with respect to the boundary $y=0$. The coefficients $S$ and $T$ are tuned so that $\widetilde \psi_\mathrm{s}$ satisfies the boundary condition \eqref{eq:local_boundary_condition_A}, and $S$ is called the scattering amplitude\footnote{In this framework, $S$ is actually a pure reflection amplitude, since the incoming wave cannot be transmitted across the sharp boundary.}. 

We are interested in the properties of $S$ as $(k_x,\kappa)$ varies. Let $U_\mathrm{out} \subset \mathbb R^2$ be an open subset and let $U_\mathrm{in}$ and $U_\mathrm{ev}$ be the images of it by the maps $(k_x,\kappa) \mapsto (k_x,-\kappa)$ and $(k_x,\kappa) \mapsto (k_x,\kappa_\mathrm{ev})$. We assume  $\widehat \psi_\mathrm{in}$, $\widehat \psi_\mathrm{out}$ and $\widehat \psi_\mathrm{ev}$ to be regular and non-vanishing on their respective domains, and also of amplitude one. Namely $\lVert \widehat \psi_\mathrm{out}(k_x,\kappa) \rVert =1$ for $k_x,\kappa \in U_\mathrm{out}$, and similarly for the others. In that case, the scattering state \eqref{eq:def_scattering_state} is uniquely defined and moreover $S(k_x,\kappa) \in U(1)$.

The relevance of the scattering amplitude comes from the following result:
\begin{thm}{\cite[Thm.~9]{GrafJudTauber21}}\label{thm:GJT21}
Let $\Gamma$ be any smooth, counter-clockwise and not self-intersecting loop in the half-plane $\{(k_x,\kappa) \in \mathbb R^2, \kappa >0\}$. Then
\begin{equation}\label{eq:buk_scattering}
\dfrac{1}{2\pi \ii}	\int_\Gamma S^{-1} \dd S = C_+
\end{equation}
Moreover,
\begin{equation}\label{eq:S_nb}
\lim_{\kappa \to 0} \int_\mathbb R \dd k_x (S^{-1}\partial_{k_x} S)(k_x,\kappa)	= n_\mathrm{b}
\end{equation}
\end{thm}

The scattering amplitude appears as a pivotal point to understand the bulk-edge correspondence. On the one hand, $S$ can be interpreted as a transition function between two bulk sections, and its winding is naturally related to the Chern number. Eq.~\eqref{eq:buk_scattering} is sometimes called the bulk-scattering correspondence. On the other hand, $\omega(k_x,\kappa)\to \omega(k_x,0)$ as $\kappa \to 0$ so that Eq.~\eqref{eq:S_nb} is the change of argument of $S$ near the bottom of the bulk spectrum. This quantity detects merging events of edge mode branches with the bulk band, thank to a relative version of Levinson's theorem, originally developped in \cite{GrafPorta13}.

\subsection{Detecting anomalies: winding number}

One sees from Theorem~\ref{thm:GJT21} above that in order for the bulk-edge correspondence to hold, the scattering amplitude needs to be regular and not winding at infinity. In that case, we can deform the loop $\Gamma$ to the real line $\kappa =0$ and recover $C_+=n_\mathrm{b}$. In order to investigate the infinite part of the spectrum, we need more on the scattering matrix. For a bulk section $\widehat \psi(k_x,k_y)$, solution of $H\widehat \psi = \omega \widehat\psi$, see \eqref{eq:bulk_eigenvalue_equation}, we denote by 
\begin{equation}\label{eq:widehatpsi}
\widehat \Psi(k_x,k_y) = \begin{pmatrix}
	\widehat \psi(k_x,k_y) \\ \ii k_y \widehat \psi(k_x,k_y)
\end{pmatrix} \in \mathbb C^4
\end{equation}
 and notice that $A \widehat\Psi$ is a column vector of size two for $A \in \mathrm{M}_{2,4}(\mathbb C)$.

\begin{prop}\label{prop:S_explicit}
	Let $\widehat \psi^0$ a bulk section that is regular and non-vanishing in $\mathbb R^2\setminus\{0\} \cup \{\infty\}$ and let $\widehat \psi^\infty$ be a bulk section that is regular and non-vanishing in $\mathbb R^2$. Let $A \in \mathrm{M}_{2,4}(\mathbb C)$ of rank 2 be a boundary condition as in \eqref{eq:local_boundary_condition_A}. The scattering amplitude reads:
	\begin{equation}\label{eq:S_explicit}
		S(k_x,\kappa) = - \dfrac{g(k_x,-\kappa)}{g(k_x,\kappa)}, \qquad g(k_x,\kappa) = \det(
			[A \widehat \Psi^0(k_x,\kappa) , A \widehat \Psi^\infty(k_x,\kappa_{\mathrm{ev}})] ),
	\end{equation}
    and is regular in the limit $k_x^2+\kappa^2 \to \infty$. In particular, equivalent boundary conditions have the same scattering amplitude $S$. 
\end{prop}

This proposition is proved in Section~\ref{sec:bulksections}. The last sentence is obvious from the expression of $g$ above and ensures that the boundary condition classification from Theorem~\ref{thm:SABC} is relevant for the investigation of the anomalous bulk-edge correspondence. To study the winding phase of $S$ near infinity, we consider the following path for $\lambda>0$ and $\delta >0$:
\begin{equation}\label{eq:defGammadelta}
\Gamma_{\delta}(\lambda) =\Big\{ \Big(k_x = \dfrac{-\lambda_x}{\lambda_x^2 + \delta^2},  \kappa = \dfrac{\delta}{\lambda_x^2 + \delta^2}\Big),  \lambda_x \in [-\lambda,\lambda] \Big\} 
\end{equation}
where $\lambda_x,\delta$ are called the dual variables to $k_x,\kappa$.

\begin{defn}
We say that $S$ winds at $\infty$ if $w_\infty\neq 0$ with 
$$
w_\infty := \lim_{\lambda \to 0} \lim_{\delta \to 0} \int_{\Gamma_{\delta}(\lambda)} S^{-1}\dd S 
$$
\end{defn}

For small $\delta$, $\Gamma_{\delta}(\lambda)$ explores the region near $k_x=0$, $\kappa = \infty$. The limit $\lambda \to 0$ ensures that we are not counting other edge modes at finite but large $k_x$. 


\begin{lem}\label{lem:C+nbomegainfty}
	 $C_+ = n_\mathrm b + w_\infty$.
\end{lem}
\begin{proof}
	According to Theorem~\ref{thm:GJT21} one has
	$$
	C_+=\dfrac{1}{2\pi \ii}	\int_\Gamma S^{-1} \dd S
	$$
	for any simple and counterclockwise loop $\Gamma$. Fix $\delta>0$ and consider 
	$$
	\mathcal C_{\delta} =\Big\{ \Big(k_x = \dfrac{-\lambda_x}{\lambda_x^2 + \delta^2},  \kappa = \dfrac{\delta}{\lambda_x^2 + \delta^2}\Big),  \lambda_x \in \mathbb R \Big\} 
	$$
	This loop naturally splits into $\mathcal C_{\delta} = \Gamma_{\delta}(\lambda) \cup (\Gamma_{\delta}(\lambda))^c$ for any $\lambda>0$, with
	$$
	(\Gamma_{\delta}(\lambda))^c =\Big\{ \Big(k_x = \dfrac{-\lambda_x}{\lambda_x^2 + \delta^2},  \kappa = \dfrac{\delta}{\lambda_x^2 + \delta^2}\Big),  \lambda_x \in]\lambda,+\infty[ \cup ]-\infty,\lambda[\Big\}
	$$
	In particular, $(\Gamma_\delta(\lambda))^c$ becomes the  real line $\kappa =0$, $ k_x \in \mathbb R$  as $\delta\to 0$ and $\lambda \to 0$, and $S$ is regular in the upper half-plane $\kappa >0$ so that 
	$$
	\lim_{\lambda \to 0} \lim_{\delta \to 0} \int_{(\Gamma_{\delta}(\lambda))^c} S^{-1}\dd S = \lim_{\kappa \to 0} \int_\mathbb R \dd k_x (S^{-1}\partial_{k_x} S)(k_x,\kappa)	= n_\mathrm{b}
	$$
	The remaining part of the integral along $\Gamma_{\delta}$ leads to $w_\infty$.
 \end{proof}

Consequently, $w_\infty$ detects when a boundary condition is anomalous at infinity, but also predicts the number of edge modes below the upper bulk band by $n_\mathrm{b}=C_+-w_\infty$ since $C_+=1$ is fixed. The main idea to compute $w_\infty$ is to expand $S$ and $g$ from Proposition~\ref{prop:S_explicit} near $k_x^2+\kappa^2 \to \infty$, then extract the leading term $\mathcal S$ along the path $\Gamma_{\delta}(\lambda)$. Ultimately, the winding number is reduced to 
$$
w_\infty = \int_{-\infty}^{+\infty} \mathcal S^{-1}(u) \mathcal S'(u) \dd u, \qquad \mathcal S(u):= \dfrac{G_-(u)}{G_+(u)}.
$$
with $G_\pm : \mathbb R \to \mathbb C^*$ two complex curves whose winding phases can be computed explicitly.

\begin{exmp}\label{example_continuated}
	We continue Example~\ref{example}. 
	\begin{itemize}
		\item Dirichlet boundary condition: one has $G_+=G_-=1$ so that $w_\infty =0$.
		\item Condition a: one has $G_-(u) = -u -\ii$ and $G_+(u)=G_-^*(u)$. As $u$ goes from $-\infty$ to $+\infty$, $G_-$ goes from $+\infty-\ii$ to $-\infty -\ii$ so that its change of argument is $-\pi$. Consequently $G_+$ argument change is $\pi$, so that $S$ winds by $-2\pi$ and $w_\infty=-1$. We recover $n_\mathrm{b}=1-(-1)=2$.
		\item Condition b: one has
		$$
		G_-(u) = -15 u^2 + 4u \sqrt{2u^2+1} + \ii \left(\sqrt{2u^2+1}-4u\right), \qquad G_+(u)=G^*_-(u).
		$$
		A detailed analysis of this complex curve (see Section~\ref{sec:anomalyA34}) shows that it is dominated by its real part as $u \to \pm \infty$. $\Re(G_-)$ goes from $-\infty$ (with $\Im(G_-)>0$) to $-\infty$ (with $\Im(G_-)<0$). Its imaginary part is strictly decreasing and vanishes at $u_0$ with $\Re(G_-(u_0))>0$. Thus $G_-$ winds once clockwise around zero as $u \in \mathbb R$, and $G_+=G_-^*$ winds once anti-clockwise. Thus $w_\infty=-2$, and we recover $n_\mathrm{b}=3$.
	\end{itemize}
\end{exmp}

What is remarkable for this problem is that $G_\pm$, $\mathcal S$ and $w_\infty$ can be computed not only for specific examples, but actually for any boundary condition from Table~\ref{tab:self-adjoint_classes}. Moreover, among the many free parameters in each class, only a few of them survive near infinity. The anomalous nature (and the value of $w_\infty$) is driven by one or two parameters only, so that an exhaustive classification can be formulated in a rather compact way. This is the central result of this paper. 

\begin{thm}\label{thm:classification}
	$w_\infty$ can be systematically computed for almost every self-adjoint boundary condition. Its value is given in Table~\ref{tab:anomalies}. Each   expression for $\mathcal S$ can be found in Section~\ref{sec:anomaly_classification}.
\end{thm}
\begingroup
\setlength{\tabcolsep}{20pt}
\renewcommand{\arraystretch}{1.2}
\begin{table}[htb]
	\centering
	\begin{tabular}{|c|c|c|}
		\hline
		Class &	$w_\infty$ & Condition	\\
		\hline
		\hline

		$\mathfrak{A}_{1,2}$ &	$0$	& None	 \\

		\hline

		$\mathfrak{A}_{1,4}$ & $\sign_0(\beta)$	&	None \\
			\hline  
		
		 &  $-\sign(\beta)$ & $\beta \neq 0$ \\ 	
		
		$\mathfrak{A}_{2,3}$ & $1$ & $\beta = 0$ and $|\alpha -\epsilon^{-1}| <1/\sqrt{2}$ \\
		& $0$ & $\beta = 0$ and $|\alpha -\epsilon^{-1}| >1/\sqrt{2}$ \\
		\hline  
		$\mathfrak{A}_{2,4}$ &	$\sign(\beta)$ & $\beta |a_{11}|^2 > \sqrt{2}$ or $\beta |a_{11}|^2 <- \sqrt{2}$\\ 	
		& $0$& $-\sqrt{2}<\beta |a_{11}|^2 < \sqrt{2}$ \\
		
		\hline  
		
		& $\sign(B_+)$	& $b_{12} \neq 0$ and $B_+ B_- <0$  \\ 	
		$\mathfrak{A}_{3,4}$  & $0$ & $b_{12} \neq 0$ and $B_\pm>0$\\
		 &$2 \times \sign(\sqrt{2}-\beta_1)$ & $b_{12} \neq 0$ and $B_\pm<0$\\
		& $\sign_0(\beta)$ & $b_{12}=0$ and  $\beta_1^2\neq 2$\\
		\hline  	\hline  
		$\mathfrak B$ 	& $0$ & None \\ 	\hline  	\hline  
		
		$\mathfrak C$ & $1$	& $a_2=0$ \\
		& $0$ & $a_2 \neq 0$\\
		\hline
	\end{tabular}
	\caption{Anomaly classification for any self-adjoint boundary condition. The class and parameters refer to Table~\ref{tab:self-adjoint_classes}. For class $\mathfrak{A}_{3,4}$ we also define $B_{\pm} :=\beta_2(\beta_1\pm\sqrt{2}) +|b_{12}|^2$. We distinguish the sign function $\sign : \mathbb R\setminus\{0\} \to \{\pm 1\}$ from its extended version $\sign_0: \mathbb R \to \{0,\pm 1\}$ which sends $0$ to $0$.   \label{tab:anomalies}}
\end{table}
\endgroup

\begin{rem}\label{rem:mainthm} 
	Beyond the classification in itself, several important consequences can be inferred:
	\begin{itemize}
		\item Class $\mathfrak{A}_{1,2}$ and $\mathfrak{B}$ are never anomalous: $w_\infty=0$, whereas all other classes are of both nature with $w_\infty \in \{0,\pm1,\pm2\}$. By Lemma~\ref{lem:C+nbomegainfty}, we deduce that $n_\mathrm{b} \in \{-1,0,1,2,3\}$.
		\item Among the many parameters in each class, only a few (up to three) actually drive the value of $w_\infty$.
		 Thus each phase is rather stable: in both cases (anomalous or not), several free parameters can vary without changing $w_\infty$. In particular, anomalies are not fine-tuned effects among the boundary condition, but are as stable as the non-anomalous regime.
		\item The possibility that $w_\infty =\pm 2$ was unexpected from previous works and could be revealed only through the full classification.
		\item There is another possibility for the upper band to be anomalous, with $w_\infty=0$ and an edge mode branch merging exactly at infinity, see \cite{GrafJudTauber21} for an example. The detection of asymptotic edge modes is another distinct classification that we postpone to future work. 
		\item Class $\mathfrak{A}_{1,2}, \mathfrak{A}_{1,4}, \mathfrak{B}$ and $\mathfrak{C}$ are completely classified. Some threshold cases remain: in $\mathfrak{A}_{2,3}$ ($\alpha - \epsilon^{-1} = \pm 1/\sqrt{2}$), $ \mathfrak{U}_{2,4}$ ($\beta |a_{11}|^2 =\pm\sqrt{2}$) and $\mathfrak{U}_{3,4}$ ($B_+=0$ or $B_-=0$), where we suspect $w_\infty$ to be ill-defined. In that sense, the classification is almost exhaustive. 
	\end{itemize}
\end{rem}

The rest of the paper is organized as follows. In Section~\ref{sec:proof_sabc} we prove Theorem~\ref{thm:SABC} and establish Table~\ref{tab:self-adjoint_classes} of self-adjoint boundary conditions. In Section~\ref{sec:asym_S} we prove Proposition~\ref{prop:S_explicit} and expand the scattering amplitude $S$ near infinity. Then we develop a general formalism to extract the dominant scale $\mathcal S$ of $S$. There, the main technical results are Proposition~\ref{lem:dominant_scale} and \ref{lem:dominant_scale2}. Finally, we apply this strategy to each class in Section~\ref{sec:anomaly_classification}. Due to the large number of parameters in some classes, some expression for $S$ are very large and tedious to manipulate (up to 60 terms, see e.g. \eqref{eq:longuest_g0}). Thus, some computation where assisted with a symbolic computation software (Mathematica), but without any numerical approximation. We also give a detailed example in Section~\ref{sec:casestudy} where all computations can be done by hand to facilitate the reading. 

%

\section{Self-adjoint boundary conditions \label{sec:proof_sabc}}

In this section we prove Theorem~\ref{thm:SABC} and establish Table~\ref{tab:self-adjoint_classes}. We adapt the method from Appendix~B of \cite{GrafJudTauber21} for shallow-water wave operator to the Dirac Hamiltonian. We study the edge operator $H^\sharp$ from \eqref{eq:edge_Hamiltonian} associated to states $\tilde\psi(y;k_x)$ from \eqref{eq:states_halfspace}. In this section, we drop the $k_x$ dependence as well as $\sharp$ and $\tilde{}$ notations, so that we work instead with $H$ and $\psi(y)$.

\begin{lem}\label{lem:Omega}
	Self-adjoint realizations of the Hamiltonian H correspond to subspaces $M\subset \mathbb{C}^4$ with
	\begin{equation}\label{eq:lem_sa_bdry_cond}
		\Omega M = M^\perp\,,
	\end{equation}
	where 
	\begin{equation}\label{eq:def_omega_thm}
		\Omega = \begin{pmatrix}
			0 & 1 & \epsilon & 0\\
			-1 & 0 & 0 & -\epsilon\\
			-\epsilon & 0 & 0 & 0\\
			0 & \epsilon & 0 & 0
		\end{pmatrix}\,.
	\end{equation}
	The domain $\mathcal{D}(H)$ is a subspace of the Sobolev space $H^2\oplus H^2$ which is characterized by
	\begin{equation}\label{eq:Psi_boundary_values}
		 \mathcal{D}(H) := \left\{\psi \in H^2\oplus H^2\, | \, \Psi = \begin{pmatrix}
		 	\psi(0)\\
		 	\psi'(0) 
		 \end{pmatrix} \in M\right\}.
	\end{equation}
\end{lem}
\begin{proof}
	Let $\psi,\phi \in H^2\oplus H^2$. A partial integration of $\braket{\phi}{H\psi}$ yields a boundary term at $y=0$:
	\begin{align}\label{eq:psisHpsi-Hpsispsi}
		\braket{\phi}{H\psi}-\braket{H\phi}{\psi} &= \phi_1^*\psi_2-\phi_2^*\psi_1 + \epsilon \left(\phi_1^*\psi_1'-\phi_2^*\psi_2'-\psi_1(\phi_1^*)'+\psi_2(\phi_2^*)'\right)\\
		&= \Phi^*\Omega\Psi\,,\label{eq:psisHpsi-Hpsispsi_f}
	\end{align}
	with $\Psi$ and $\Phi$ as in~\eqref{eq:Psi_boundary_values}, $\Omega$ as in~\eqref{eq:def_omega_thm}. In order for $H$ to be self-adjoint, we require $\Psi \in M \subset \mathbb C^4$ with the following properties
	\begin{enumerate}[i)]
	\item $\Phi^*\Omega\Psi = 0$,~$\Psi\in M$ $\Rightarrow$~$\Phi\in M$,
	\item $\Psi\in M$ $\Rightarrow$~$\Phi^*\Omega\Psi=0$ for all $\Phi\in M$.
\end{enumerate}
By~\eqref{eq:psisHpsi-Hpsispsi_f}, the properties (i) and (ii) imply $H^*\subset H$ and $H\subset H^*$ respectively, and thus $H=H^*$.

Moreover these two properties imply $(\Omega M)^\perp=M$. Indeed, let $\Phi \in (\Omega M)^\perp$ and $\Psi \in M$. One has $\Phi^* \Omega \Psi=0$ so that $\Phi \in M$ by (i). Thus $(\Omega M)^\perp \subset M$. Conversely, let $\Psi \in M$. For any $\Phi \in M$ one has $\Phi^* \Omega \Psi =0$ by (ii), which implies $\Psi^* \Omega \Phi =0$ since $\Omega^*=-\Omega$. Thus $\Psi \in (\Omega M)^\perp$ and $M \subset (\Omega M)^\perp$. 

Reciprocally, it is not hard to check that $(\Omega M)^\perp=M$ implies (i) and (ii), so that $H$ is self-adjoint if and only if $(\Omega M)^\perp=M$, namely $\Omega M = M^\perp$.
\end{proof}
\begin{rem}
	The matrix $\Omega$ defines a symplectic form on $\mathbb C^4$ and the subspace $M$  is sometimes called a Lagrangian space, a general structure that  naturally appears for self-adjoint extensions of elliptic operators, see e.g. \cite{Gontier23}.
\end{rem}

\begin{prop}\label{prop:sa_bdry_cond}
	\begin{enumerate}[i)]
		\item Subspaces $M\subset\mathbb{C}^4$ with $\dim M =2$ are determined by $2\times 4$ matrices $A$ of maximal rank, i.e.~$\rk A =2$ by means of
		\begin{equation}\label{eq:prop_M_kerA}
			M = \{\Psi \in \mathbb{C}^4\lvert A\Psi =0\}=\ker A\,,
		\end{equation}
		and conversely. Two such matrices $A$, $\widetilde{A}$ determine the same subspace if and only if $A=B\widetilde{A}$ with $B\in\mathrm{GL}_2(\mathbb C)$.
		\item Self-adjoint boundary conditions as in Lemma~\ref{lem:Omega} are characterized by matrices $A$ as in~i) with
		\begin{equation}\label{eq:prop_AOAs}
			A\Omega^{-1}A^*=0\,.
		\end{equation}
	\end{enumerate}
\end{prop}%
\begin{proof}
	\begin{enumerate}[i)]
		\item The first part is clear as $A\Psi =0$ gives two conditions if $A$ is of $\rk$ 2. In other words $\Omega$ is invertible and thus $\ker\Omega=\{0\}$. For the second part (the last sentence) suppose that $A$ and $\widetilde{A}$ determine the same subspace, then there is a map $B: \mathbb{C}^2\to\mathbb{C}^2$, which is well-defined by $Av\mapsto \widetilde{A}v$, $v\in \mathbb{C}^4$ as $\ker A = \ker \widetilde A$. If the kernel is to be preserved then $B\in\mathrm{GL}(2)$.
		\item Self-adjoint boundary conditions correspond to subspaces $M\subset\mathbb{C}^4$ with $\Omega M = M^\perp$. 
		Moreover, $\Omega$ is invertible so that $\Omega M=M^\perp$ is equivalent to  $\Omega^{-1}M^\perp \subset M$ and $\dim M =2$. Thus, by (i), $M=\ker A$ for some $A \in \mathrm{M}_{2,4}(\mathbb C)$ of rank 2. Moreover one has $M^\perp=\ran A^* =\{A^*v\lvert v\in\mathbb{C}^2\}$. Let $v_1, v_2 \in \mathbb C^2$. One has
		$$
		\langle v_1, A\Omega^{-1} A^* v_2 \rangle = \langle A^* v_1, \Omega^{-1} A^* v_2 \rangle =0
		$$
		since $ A^* v_1 \in M^\perp$ and $\Omega^{-1} A^* v_2 \in \Omega^{-1}M^\perp \subset M$. 
		
		Reciprocally, if $M = \ker A$ with $A \in \mathrm{M}_{2,4}(\mathbb C)$ of rank 2 then $\dim M=2$ and one can check that $\Omega^{-1}M^\perp \subset (M^\perp)^\perp =M$.
	\end{enumerate}
\end{proof}

\subsection{Grassmannian}

Self-adjoint boundary conditions are given by $2\times 4$ matrices of rank 2 up to $\mathrm{GL}_2$-invariance. This set of matrices corresponds to the Grassmannian $\mathrm{Gr}(4,2)$, which parametrizes $2$-dimensional subspaces in the $4$-dimensional vector space $\mathbb{C}^4$:
\begin{equation}
	\mathrm{Gr}(4,2)\cong \frac{\{A \in \mathrm{M}_{2,4}(\mathbb C) \, | \, \rk(A) =2 \}}{\mathrm{GL}_2(\mathbb C)}\,.
\end{equation}
Grassmannians are complex manifolds for which a lot is known, see e.g.~\cite{GriffithsHarris14}. We shall use the following facts: since $\rk(A)=2$ there are at least 2 columns among the 4 of $A \in \mathrm{M}_{2,4}(\mathbb C)$  which are linearly independent vectors in $\mathbb C^2$. Then one can use the $\mathrm{GL}_2(\mathbb C)$ invariance to set the first one to $\begin{pmatrix}
	1 \\0 
\end{pmatrix}$ and the second to $\begin{pmatrix}
0 \\1 
\end{pmatrix}$ and reduce the matrix $A$ to its row echelon form (from the right). This provides the following partition 
\begin{equation}
	\mathrm{Gr}(4,2) = \bigsqcup_{\bvec{j}}\mathcal{C}_{\bvec{j}}\,.
\end{equation}
where $\bvec{j}=\{1,2\}$, $\{1,3\}$, $\{1,4\}$, $\{2,3\}$, $\{2,4\}$ and $\{3,4\}$, with 
\begin{equation}\label{eq:Schubert_cells}
	\begin{aligned}
		\mathcal{C}_{\{1,2\}}: \begin{pmatrix}
			1 & 0 & 0 & 0\\
			0 & 1 & 0 & 0
		\end{pmatrix}\,,\quad &\mathcal{C}_{\{1,3\}}: \begin{pmatrix}
			1 & 0 & 0 & 0\\
			0 & \star & 1 & 0
		\end{pmatrix}\,,\quad \mathcal{C}_{\{1,4\}}: \begin{pmatrix}
			1 & 0 & 0 & 0\\
			0 & \star & \star & 1
		\end{pmatrix}\,,\\
		\mathcal{C}_{\{2,3\}}: \begin{pmatrix}
			\star & 1 & 0 & 0\\
			\star & 0 & 1 & 0
		\end{pmatrix}\,,\quad &\mathcal{C}_{\{2,4\}}: \begin{pmatrix}
			\star & 1 & 0 & 0\\
			\star & 0 & \star & 1
		\end{pmatrix}\,,\quad \mathcal{C}_{\{3,4\}}: \begin{pmatrix}
			\star & \star & 1 & 0\\
			\star & \star & 0 & 1
		\end{pmatrix}\,.
	\end{aligned}
\end{equation}
The $\mathcal{C}_{\bvec{j}}$ are called the Schubert cells of $\mathrm{Gr}(4,2)$. Other Schubert cells of $\mathrm{Gr}(4,2)$ exist but the six above are enough to describe each element uniquely. Notice that $\dim(\mathcal{C}_{\bvec{j}})=\sum_{i=1}^2 j_i -i$ and $\mathcal{C}_{\bvec{j}}\cong \mathbb{C}^{\dim(\mathcal{C})}$ so that the Schubert cells above have complex dimensions $0,1,2,2,3$ and $4$, respectively.

\subsection{Local boundary conditions} According to \eqref{eq:local_boundary_condition} we are interested in local boundary condition of the form $A \Psi =0$ with $A = A_0+\ii k_x A_1 \in \mathrm{M}_{2,4}(\mathbb C)$. We shall start with
\begin{equation}
	A_0 = \begin{pmatrix}
		a_{11} & a_{12} & a_{13} & a_{14} \\ a_{21} & a_{22} & a_{23} & a_{24}
	\end{pmatrix} = [A_1 \quad A_2], \qquad 
	A_1 = \begin{pmatrix}
		b_{11} & b_{12} & 0 & 0 \\ b_{21} & b_{22} & 0 & 0
	\end{pmatrix} = [B \quad 0]
\end{equation}
with $A_1, A_2$ and $B \in \mathrm{M}_2( \mathbb C)$.  According to Proposition~\ref{prop:sa_bdry_cond}, $A$ has to be a rank-2 matrix. Thus we consider the various cases with respect to the rank of $A_0$ and $A_1$. Depending on their rank, such matrices can be simplified further using the $\mathrm{GL}_2(\mathbb C)$-invariance. 

Moreover, $A$ satisfies $$A\Omega^{-1} A^* = 0$$ with 
\begin{equation}\label{eq:omeginv}
	\Omega^{-1} = \epsilon^{-2} \begin{pmatrix}
		0 & 0 & -\epsilon & 0\\
		0 & 0 & 0 & \epsilon\\
		\epsilon & 0 & 0 & -1\\
		0 & -\epsilon & 1 & 0
	\end{pmatrix}:= \epsilon^{-2} \begin{pmatrix}
	0 & -\Omega_1 \\  \Omega_1 & \Omega_2
\end{pmatrix}\,.
\end{equation}
with $\Omega_1,\Omega_2 \in \mathrm{M}_2(\mathbb C)$ and $\Omega_1^*=\Omega_1$ and  $\Omega_2^*=\Omega_2$. The condition $A\Omega^{-1} A^* = 0$ becomes
$$
- A_1 \Omega_1 A_2^* + A_2 \Omega_1 A_1^* + A_2 \Omega_2 A_2^* - \ii k_x (B \Omega_1 A_2^* + A_2 \Omega_1 B^*) = 0.
$$
This relation has to be valid for every $k_x \in \mathbb R$ so we infer
\begin{align}\label{eq:sa_constraint2x2}
	 - A_1 \Omega_1 A_2^* + A_2 \Omega_1 A_1^* + A_2 \Omega_2 A_2^* =0, \qquad 
	 B \Omega_1 A_2^* + A_2 \Omega_1 B^* =0.
\end{align}

\subsection{\texorpdfstring{Class $\mathfrak A$: $\rk(A_0)=2$}{Clas A: rk(A0)=2}} 

In this section we assume $\rk(A_0)=2$ and $A_1$ arbitrary:
$$	A_1 = \begin{pmatrix}
	b_{11} & b_{12} & 0 & 0 \\ b_{21} & b_{22} & 0 & 0
\end{pmatrix} = [B \quad 0]$$
The $\mathrm{GL}_2(\mathbb C)$-invariance allows to reduce $A_0$ to one of the six Schubert cells from \eqref{eq:Schubert_cells}. For each of them we investigate \eqref{eq:sa_constraint2x2} and possibly restrict some parameters. 

\paragraph{Class $\mathfrak{A}_{1,2}$.} In that case $A_0 \in \mathcal C_{\{1,2\}}$, namely
$$
A_0 = \begin{pmatrix}
	1 & 0 & 0 & 0\\
	0 & 1 & 0 & 0
\end{pmatrix}
$$
One can check that \eqref{eq:sa_constraint2x2} is satisfied for any $B \in \mathbb \mathrm{M}_2(\mathbb C)$. 

\paragraph{Class $\mathfrak{A}_{1,3}$.} In that case $A_0 \in \mathcal C_{\{1,3\}}$, namely
$$
A_0 = \begin{pmatrix}
	1 & 0 & 1& 0\\
	0 & a_{22} & 0 & 0
\end{pmatrix}.
$$
In that case we have 
$$
- A_1 \Omega_1 A_2^* + A_2 \Omega_1 A_1^* + A_2 \Omega_2 A_2^* = \begin{pmatrix}
	0 & -\epsilon \\\epsilon & 0
\end{pmatrix} 
$$
which never vanishes for $\epsilon >0$ so that \eqref{eq:sa_constraint2x2} is never satisfied. Thus there is no self-adjoint boundary condition in this class, this is why it is not appearing in Table~\ref{tab:self-adjoint_classes}.

\paragraph{Class $\mathfrak{A}_{1,4}$.} In that case $A_0 \in \mathcal C_{\{1,4\}}$, namely
$$
A_0 = \begin{pmatrix}
	1 & 0 & 0& 0\\
	0 & a_{22} & a_{23} & 1
\end{pmatrix}.
$$
In that case we have 
$$
- A_1 \Omega_1 A_2^* + A_2 \Omega_1 A_1^* + A_2 \Omega_2 A_2^* = \left(
\begin{array}{cc}
	0 & -\epsilon  \left(a_{23}\right){}^* \\
	a_{23} \epsilon  & -\epsilon  \left(a_{22}\right){}^*+\left(a_{23}\right){}^*+a_{22} \epsilon -a_{23} \\
\end{array}
\right),
$$
from which we infer $a_{23}=0$ and then $a_{22} \in \mathbb R$. Moreover, with that knowledge, we compute
$$
 B \Omega_1 A_2^* + A_2 \Omega_1 B^* = \left(
 \begin{array}{cc}
 	0 & -\epsilon  b_{12}  \\
 	-\epsilon  \left(b_{12}\right){}^* & -\epsilon (b_{22} +  \left(b_{22}\right){}^*) \\
 \end{array}
 \right),
 $$
 from which we infer $b_{12}=0$ and $b_{22} \in \ii \mathbb R$. Denoting $a_{22}=\alpha$ and $b_{22}= \ii \beta$ we deduce the self-adjoint matrices in that case
 $$
 A_0 = \begin{pmatrix}
 	1 & 0 & 0& 0\\
 	0 & \alpha & 0 & 1
 \end{pmatrix}, \qquad  A_1 = \begin{pmatrix}
 b_{11} & 0 & 0& 0\\
 b_{21} & \ii \beta & 0 & 0
\end{pmatrix},
 $$
 with $\alpha,\beta \in \mathbb R$ and $b_{11}, b_{21} \in \mathbb C$.

\paragraph{Class $\mathfrak{A}_{2,3}$.} In that case $A_0 \in \mathcal C_{\{2,3\}}$, namely
$$
A_0 = \begin{pmatrix}
	a_{11} & 1& 0& 0\\
	a_{21} &0& 1& 0
\end{pmatrix}.
$$
In that case we have 
$$
- A_1 \Omega_1 A_2^* + A_2 \Omega_1 A_1^* + A_2 \Omega_2 A_2^* = \left(
\begin{array}{cc}
	0 & -\epsilon a_{11} \\
	\epsilon  \left(a_{11}\right){}^* & \epsilon  \left(a_{21}\right){}^*-a_{21} \epsilon  \\
\end{array}
\right),
$$
from which we infer $a_{11} =0$ and $a_{21} = \alpha \in \mathbb R$. Moreover one has
$$
B \Omega_1 A_2^* + A_2 \Omega_1 B^*  = \left(
\begin{array}{cc}
	0 & b_{11} \epsilon  \\
	\epsilon  \left(b_{11}\right){}^* & \epsilon  \left(b_{21}\right){}^*+b_{21} \epsilon  \\
\end{array}
\right)
$$
from which we infer $b_{11} =0$ and $b_{21}=\ii \beta \in \ii \mathbb R$. The self-adjoint matrices in that case are
$$
A_0 = \begin{pmatrix}
	0 & 1 & 0& 0\\
	\alpha & 0 & 1 & 0
\end{pmatrix}, \qquad  A_1 = \begin{pmatrix}
	0 & b_{12} & 0& 0\\
	\ii \beta & b_{22} & 0 & 0
\end{pmatrix},
$$
with $\alpha,\beta \in \mathbb R$ and $b_{12}, b_{22} \in \mathbb C$.

\paragraph{Class $\mathfrak{A}_{2,4}$.} In that case $A_0 \in \mathcal C_{\{2,4\}}$, namely
$$
A_0 = \begin{pmatrix}
	a_{11} & 1& 0& 0\\
	a_{21} &0& a_{23}& 1
\end{pmatrix}.
$$
In that case we have 
$$
- A_1 \Omega_1 A_2^* + A_2 \Omega_1 A_1^* + A_2 \Omega_2 A_2^* = \left(
\begin{array}{cc}
	0 & \epsilon(1  -a_{11}   \left(a_{23}\right){}^*) \\
	-\epsilon (1- a_{23}   \left(a_{11}\right){}^*)  &   \epsilon ( a_{23} \left(a_{21}\right){}^* - a_{21}  \left(a_{23}\right){}^*)+\left(a_{23}\right){}^*-a_{23} \\
\end{array}
\right)
$$
from which we infer $1  -a_{11}   (a_{23})^* =0$, that can be rewritten $a_{11} \neq 0$ and $a_{23} = (a_{11}^{-1})^*$. The lower right coefficient of the matrix above then implies
$$
\epsilon \left(  \left(\dfrac{a_{21}}{a_{11}}\right)^* - \dfrac{a_{21}}{a_{11}}\right)+\dfrac{1}{a_{11}}-\dfrac{1}{(a_{11})^*} = 0
$$
which can be rephrased as
$$
\Im\left(\dfrac{a_{21}-\epsilon^{-1}}{a_{11}}\right)= 0.
$$
Thus we have $a_{21}= \alpha a_{11}+\epsilon^{-1}$ with $\alpha \in \mathbb R$.

The second part of \eqref{eq:sa_constraint2x2} reads
$$
B \Omega_1 A_2^* + A_2 \Omega_1 B^* = \left(
\begin{array}{cc}
	0 &\epsilon  (b_{11} a_{11}^{-1}-b_{12}) \\
	\epsilon  (b_{11} a_{11}^{-1}-b_{12})^* &  \epsilon \left( \left(b_{21}a_{11}^{-1}\right)^*+b_{21}   a_{11}^{-1}-  \left(b_{22}\right){}^*-b_{22} \right)  \\
\end{array}
\right)
$$
which implies $b_{12} = b_{11} a_{11}^{-1}$ and 
$$
\Re\left(b_{21}   a_{11}^{-1}-b_{22} \right) =0.
$$
Thus we have $b_{21}   a_{11}^{-1}-b_{22} = \ii \beta$ with $\beta \in \mathbb R$.

 The self-adjoint matrices in that case are
$$
A_0 =  \begin{pmatrix}
	a_{11} & 1& 0& 0\\
	a_{21} &0& (a_{11}^{-1})^*& 1
\end{pmatrix}, \qquad  A_1 = \begin{pmatrix}
	b_{11} & b_{11}a_{11}^{-1} & 0& 0\\
	b_{21} & b_{22} & 0 & 0
\end{pmatrix},
$$
with $a_{11} \in \mathbb C\setminus\{0\}$, $a_{21}, b_{11}, b_{21}, b_{22} \in \mathbb C$ and $a_{21}= \alpha a_{11}+\epsilon^{-1}$, $\alpha \in \mathbb R$ as well as $b_{21}   a_{11}^{-1}-b_{22} = \ii \beta$, $\beta \in \mathbb R$.

\paragraph{Class $\mathfrak{A}_{3,4}$.} In that case $A_0 \in \mathcal C_{\{3,4\}}$, namely
$$
A_0 = \begin{pmatrix}
	a_{11} & a_{12}& 1& 0\\
	a_{21} &a_{22}& 0& 1
\end{pmatrix}.
$$
In that case we have 
$$
0 = - A_1 \Omega_1 A_2^* + A_2 \Omega_1 A_1^* + A_2 \Omega_2 A_2^* = \left(
\begin{array}{cc}
	\epsilon  \left(a_{11}\right){}^*-a_{11} \epsilon  & \epsilon  \left(a_{21}\right){}^*+a_{12} \epsilon -1 \\
	-\epsilon  \left(a_{12}\right){}^*-a_{21} \epsilon +1 & a_{22} \epsilon -\epsilon  \left(a_{22}\right){}^* \\
\end{array}
\right),
$$
from which we infer $a_{11} = \alpha_1 \in \mathbb R$, $a_{22} = \alpha_2 \in \mathbb R$ and $a_{21}=\epsilon^{-1}-(a_{12})^*$. Moreover, we have
$$
0 = B \Omega_1 A_2^* + A_2 \Omega_1 B^*  = \left(
\begin{array}{cc}
	\epsilon  \left(b_{11}\right){}^*+b_{11} \epsilon  & \epsilon  \left(b_{21}\right){}^*-b_{12} \epsilon  \\
	b_{21} \epsilon -\epsilon  \left(b_{12}\right){}^* & -\epsilon  \left(b_{22}\right){}^*-b_{22} \epsilon  \\
\end{array}
\right),
$$
from which we infer $b_{11} = \ii \beta_1 \in \ii \mathbb R$, $b_{22} = \ii \beta_2 \in \ii \mathbb R$ and $b_{21}= (b_{12})^*$.

 The self-adjoint matrices in that case are
$$
A_0 =  \begin{pmatrix}
	\alpha_1 & a_{12}& 1& 0\\
	\epsilon^{-1}-(a_{12})^* &\alpha_2 & 0& 1
\end{pmatrix}, \qquad  A_1 = \begin{pmatrix}
	\ii \beta_1 & b_{12}& 0& 0\\
	(b_{12})^* & \ii \beta_2 & 0 & 0
\end{pmatrix},
$$
with $\alpha_1, \alpha_2, \beta_1, \beta_2 \in \mathbb R$ and $a_{12}, b_{12} \in \mathbb C$.

\subsection{Class $\mathfrak{B}$: $\rk(A_1) = 2$}

In this section we assume $\rk(A_1) =2$. Since $A_1 = [B\, |\, 0]$ this means that $\rk(B)=2$ and thus by the $\mathrm{GL}_2(\mathbb C)$-invariance we can reduce the study to
$$
A_1 = \begin{pmatrix}
	1 & 0 & 0& 0\\
	0 & 1 & 0 & 0
\end{pmatrix}
$$
To avoid any overlap with class $\mathfrak A$ it is sufficient to consider the case where $\rk(A_0) \leq 1$. Thus we write
$$
A_0 = \begin{pmatrix}
	a_1 & a_2 & a_3 & a_4\\
	\mu a_1 & \mu a_2 & \mu a_3 & \mu a_4
\end{pmatrix}
$$
with $a_1, a_2, a_3, a_4, \mu \in \mathbb C$. 

In that case we have 
$$
0 = B \Omega_1 A_2^* + A_2 \Omega_1 B^* = \left(
\begin{array}{cc}
	\epsilon  \left(a_3\right){}^*+a_3 \epsilon  & \epsilon  \left(\mu  a_3\right){}^*-a_4 \epsilon  \\
	a_3 \mu  \epsilon -\epsilon  \left(a_4\right){}^* & -\epsilon  \left(\mu  a_4\right){}^*-a_4 \mu  \epsilon  \\
\end{array}
\right),
$$
from which we infer $a_3 = \ii \alpha \in \ii \mathbb R$ and $a_4 = \ii \alpha \mu^*$ (in particular $\mu a_4 \in \ii \mathbb R$). Moreover, we have
\begin{align*}
0& = - A_1 \Omega_1 A_2^* + A_2 \Omega_1 A_1^* + A_2 \Omega_2 A_2^* \cr &= \alpha \left( \alpha (\mu^*-\mu) + \ii \epsilon \alpha (a_1^*+a_1 - a_2 \mu - a_2^*\mu^*)\right) \begin{pmatrix}
1 & \mu^* \\ \mu  & |\mu|^2
\end{pmatrix},
\end{align*}
from which we infer 
$$
\alpha \Big(\alpha \Im(\mu) - \epsilon \Re(a_1-a_2\mu)\Big)=0. 
$$

The self-adjoint matrices in that case are
$$
A_0 = \begin{pmatrix}
	a_1 & a_2 & \alpha & \alpha \mu^* \\
	\mu a_1 & \mu a_2 & \alpha \mu & \alpha |\mu|^2
\end{pmatrix}, \qquad A_1 = \begin{pmatrix}
1 & 0 & 0& 0\\
0 & 1 & 0 & 0
\end{pmatrix},
$$
with $a_1, a_2, \mu \in \mathbb C$, $\alpha \in \mathbb R$ and $\alpha \big(\alpha \Im(\mu) - \epsilon \Re(a_1-a_2\mu)\big)=0$. This last condition is equivalent to one of the three cases:
\begin{enumerate}
	\item $\alpha =0$
	\item $\alpha \neq 0,\, \mu \in \mathbb R, \, a_1 =a_2\mu + \ii\beta, \, \beta \in \mathbb R$
	\item $ \mu \in \mathbb C \setminus \mathbb R$ and $\alpha=\tfrac{\epsilon}{\Im(\mu)}  \Re(a_1-a_2\mu) \in \mathbb R$
\end{enumerate}
but we shall not use them explicitly below so we keep the general constraint instead.

\subsection{Class $\mathfrak{C}$: $\rk(A_0) = \rk(A_1) =1 $}

In this section we consider the case where $A_0$ and $A_1$ are exactly of rank 1, with $\rk(A_0+\ii k_x A_1)=2$.  Since $A_1 = [B\, |\, 0]$ this means that $\rk(B)=1$  and thus by the $\mathrm{GL}_2(\mathbb C)$-invariance we can reduce the study to
$$
A_1 = \begin{pmatrix}
	1 & 0 & 0& 0\\
	0 & 0 & 0 & 0
\end{pmatrix}
$$
Then, since $\rk(A_0)=1$ we write
$$
A_0 = \begin{pmatrix}
	a_1 & a_2 & a_3 & a_4\\
	\mu a_1 & \mu a_2 & \mu a_3 & \mu a_4
\end{pmatrix}
$$
with $a_1, a_2, a_3, a_4, \mu \in \mathbb C$. Moreover one has $\mu \neq 0$ and $(a_2,a_3,a_4)\neq 0$, otherwise $\rk(A_0 + \ii k_x A_1)=1<2$.

In that case we have 
$$
0 = B \Omega_1 A_2^* + A_2 \Omega_1 B^* = 
\left(
\begin{array}{cc}
	\epsilon  \left(a_3\right){}^*+a_3 \epsilon  & \epsilon  \left(\mu  a_3\right){}^* \\
	a_3 \mu  \epsilon  & 0 \\
\end{array}
\right)
$$
from which we infer $a_3=0$. Moreover, we have
\begin{align*}
	0& = - A_1 \Omega_1 A_2^* + A_2 \Omega_1 A_1^* + A_2 \Omega_2 A_2^* = \epsilon  \left(a_2 \left(a_4\right){}^*-a_4 \left(a_2\right){}^*\right)\left(
	\begin{array}{cc}
		1 & \mu ^* \\
		\mu  & |\mu|^2   \\
	\end{array}
	\right)
\end{align*}
from which we infer $\Im(a_2 a_4^*) = 0$. 

The self-adjoint matrices in that case are
$$
A_0 = \begin{pmatrix}
	a_1 & a_2 & 0 & a_4 \\
	\mu a_1 & \mu a_2 & 0 & \mu a_4
\end{pmatrix}, \qquad A_1 = \begin{pmatrix}
	1 & 0 & 0& 0\\
	0 & 0 & 0 & 0
\end{pmatrix},
$$
with $a_1, a_2, a_4, \mu \in \mathbb C$ such that $\mu \neq 0$, $(a_2,a_4)\neq 0$ and $\Im(a_2 a_4^*) = 0$. 

\section{Asymptotic expansion of the scattering amplitude \label{sec:asym_S}}

\subsection{Bulk eigensections \label{sec:bulksections}}

Here we prove Proposition~\ref{prop:S_explicit} and also provide explicit expression for bulk sections which are used further below. Boundary condition \eqref{eq:local_boundary_condition_A} reads
\begin{equation}
	A\Psi|_{y=0} = 0\,,\quad \Psi = \begin{pmatrix}
		\widetilde \psi\\
		\widetilde \psi'
	\end{pmatrix},
\end{equation}
 for $A \in \mathrm{M}_{2,4}(\mathbb C)$. 
We apply it to the scattering state from \eqref{eq:def_scattering_state} at $y=0$, 
\begin{align}
&\widetilde\psi_s = \widehat \psi_\mathrm{in}+S\widehat\psi_\mathrm{out}+T\widehat\psi_\mathrm{ev}\cr
& \widetilde\psi_s^\prime = -\ii \kappa \widehat \psi_\mathrm{in}+\ii \kappa S\widehat\psi_\mathrm{out}+ \ii \kappa_\mathrm{ev} T\widehat\psi_\mathrm{ev}
\end{align}
where $\widehat \psi_\mathrm{in/out/ev}$ are bulk solution at momentum $k_x$ and $k_y=-\kappa,\kappa,\kappa_{\mathrm{ev}}$ respectively.  Thus at $y=0$ we have
$$
\Psi_s  = \widehat \Psi_\mathrm{in} + S \widehat \Psi_\mathrm{out} + T  \widehat \Psi_\mathrm{ev}
$$
as in \eqref{eq:widehatpsi}. This amounts to
\begin{equation}
	A\Psi_s = \begin{pmatrix}
		-\, a_1\, -\\
		-\,a_2\,-
	\end{pmatrix}\Psi_s = \begin{pmatrix}
		a_1\cdot \Psi_s\\
		a_2\cdot \Psi_s
	\end{pmatrix} =\begin{pmatrix}
		0\\ 
		0
	\end{pmatrix}\,,
\end{equation}
where $a_i \in M_{1,4}(\mathbb C)$ is the $i$-th row of $A$ for $i=1,2$. This means that we get a $2\times 2$ linear system of equations to be solved for $S$ and $T$
\begin{align}
	a_1\cdot(\widehat \Psi_\mathrm{in}+S\widehat\Psi_\mathrm{out}+T\widehat\Psi_\mathrm{ev}) &= 0\,,\\
	a_2\cdot(\widehat\Psi_\mathrm{in}+S\widehat\Psi_\mathrm{out}+T\widehat\Psi_\mathrm{ev}) &= 0\,.
\end{align}
Its solution for $S$ is 
\begin{equation}
	S = -\frac{\begin{vmatrix}
			a_1\cdot \widehat\Psi_\mathrm{in} & a_1\cdot \widehat\Psi_\mathrm{ev}\\
			a_2\cdot \widehat\Psi_\mathrm{in} & a_2\cdot \widehat\Psi_\mathrm{ev}
	\end{vmatrix}}{\begin{vmatrix}
			a_1\cdot \widehat\Psi_\mathrm{out} & a_1\cdot \widehat\Psi_\mathrm{ev}\\
			a_2\cdot \widehat\Psi_\mathrm{out} & a_2\cdot \widehat\Psi_\mathrm{ev}
	\end{vmatrix}} = - \frac{\begin{vmatrix}
	A \widehat\Psi_\mathrm{in} & A\widehat \Psi_\mathrm{ev}
\end{vmatrix}}{\begin{vmatrix}
A \widehat\Psi_\mathrm{out} & A \widehat \Psi_\mathrm{ev}
\end{vmatrix}}
\end{equation}
Let $\widehat \psi^0$ a bulk section that is regular and non-vanishing in $\mathbb R^2\setminus\{0\} \cup \{\infty\}$ and let $\widehat \psi^\infty$ be a bulk section that is regular and non-vanishing in $\mathbb R^2$. Then for $k_x \in \mathbb R$ and $\kappa >0$ we set $\widehat \psi_\mathrm{in}(k_x,-\kappa) = \widehat\psi_0(k_x,-\kappa)$, $\widehat \psi_\mathrm{out}(k_x,\kappa) = \widehat\psi_0(k_x,\kappa)$ and $\widehat\psi_\mathrm{ev}(k_x,\kappa_\mathrm{ev}) =  \widehat\psi_\infty(k_x,\kappa_\mathrm{ev})$, which leads to the expression \eqref{eq:S_explicit} for $S$ and $g$ from Proposition~\ref{prop:S_explicit}.

To expand $S$ near $\infty$ we provide explicit expressions for the sections $\widehat\psi_0$ and $\widehat\psi_\infty$. We concentrate on the upper part of the spectrum and only discuss the eigensections corresponding to $\omega_+(k)=+\sqrt{k^2+(m-\epsilon k^2)^2}$. Assuming $m$, $\epsilon>0$ the corresponding Chern number is $C_+=1$, and therefore it is impossible to find a global bulk eigensection $\widehat{\psi}$ regular for all points on the compactified $k$-plane, $\mathbb{C}\cup\{\infty\}\cong S^2$. To cover the sphere we need at least two distinct ones overlapping in a region and regular where they are defined. Explicitly, one section satisfying $H\widehat{\psi}=\omega_+\widehat{\psi}$ is given by
\begin{equation}\label{eq:bulk_eigensection_zero}
	\widehat\psi_0(k_x,k_y) = \frac{1}{\sqrt{2}\sqrt{1-q(\bvec{k})}}\begin{pmatrix}
		\frac{k_x-\ii k_y}{\omega_+(k)}\\
		q(\bvec{k})-1
	\end{pmatrix}\,,\quad q(\bvec{k}):=\frac{m-\epsilon \bvec{k}^2}{\omega_+}\,.
\end{equation}
One can check that $\norm{\widehat\psi_0(k_x,k_y)}=1$. Notice that $\omega_+(k)\sim \epsilon k^2$ and $q\to -\epsilon/\abs{\epsilon}=-1$ as $k \to\infty$. On the other hand, $\omega_+(k) \to m$ and
$$
q(k) = 1  - \dfrac{k^2}{2m^2} + o(k^2)
$$
as $k \to 0$. Thus, writing $k_x+\ii k_y = k \ee^{\ii\phi}$, 
\begin{equation}\label{eq:bulk_eigensection_used_zero}
	\lim_{k\to 0} \widehat\psi_0(k_x,k_y) = 		\begin{pmatrix}
		\ee^{-\ii\phi}\\
		0
	\end{pmatrix}, \qquad 	\lim_{k\to \infty} \widehat\psi_0(k_x,k_y) = 		\begin{pmatrix}
	0\\
	-1
\end{pmatrix}\,.
\end{equation}
Therefore, $\psi_0$ is regular at $k=\infty$ but not at $k=0$. There, it has a removable phase singularity up to a gauge transformation. We define
$
\psi_\infty = \lambda\psi_0$ with $ \lambda = \ee^{\ii\phi}=k^{-1}(k_x+\ii k_y)
$, or explicitly
\begin{equation}\label{eq:bulk_eigensection_infty}
	\widehat\psi_\infty(k_x,k_y) = \frac{1}{k\sqrt{2}\sqrt{1-q(\bvec{k})}}\begin{pmatrix}
		\frac{k^2}{\omega_+(k)}\\
		(k_x+\ii k_y)(q(\bvec{k})-1)
	\end{pmatrix}\,.
\end{equation}
In particular, $\psi_\infty$ is regular at $k=0$ but not at $k=\infty$:
$$
	\lim_{k\to 0} \widehat\psi_\infty(k_x,k_y) = 		\begin{pmatrix}
		1\\
		0
	\end{pmatrix}, \qquad 	\lim_{k\to \infty}\widehat \psi_\infty(k_x,k_y) = 		\begin{pmatrix}
		0\\
		-\ee^{\ii \phi}
	\end{pmatrix}\,.
$$

Now for $k_x \in \mathbb R$ and $\kappa >0$ consider $k_y=\kappa_{\mathrm{ev}}(k_x,\kappa) \in \ii \mathbb R$ , see \eqref{eq:kappa_ev_div}. Define
$$
\tilde k^2 = k_x^2+\kappa_{\mathrm{ev}}^2 = -(k_x^2 + \kappa^2) - \dfrac{1-2 \epsilon m}{\epsilon^2}<0,
$$
as well as $\tilde \omega_+ = \sqrt{\tilde k^2 + (m-\epsilon \tilde k^2)^2}$ and $\tilde q = (m-\epsilon \tilde k^2)\tilde \omega_+^{-1}$. One can check that $\tilde q >1$ so that $\tilde k^2(1-\tilde q)>0$ and thus $\tilde k\sqrt{1-\tilde q}$ is well defined, and so is $\widehat \psi_\infty(k_x,\kappa_{\mathrm{ev}}(k_x,\kappa))$. In particular 
$$
\lim_{(k_x,\kappa)\to \infty} \tilde q = 1, \qquad \lim_{(k_x,\kappa)\to \infty} \tilde k \sqrt{1-\tilde q} = \dfrac{1}{\sqrt{2} \epsilon},
$$
so that
$$
\lim_{(k_x,\kappa)\to \infty} \widehat\psi_\infty(k_x,\kappa_{\mathrm{ev}}(k_x,\kappa)) = \begin{pmatrix}
	1\\
	0
\end{pmatrix}\,.
$$ 
Notice that $\kappa_{\mathrm{ev}} \to \ii \epsilon^{-1}\sqrt{1-2\epsilon m} \neq 0$ as $(k_x,\kappa)\to 0$ so that $ \widehat\psi_\infty(k_x,\kappa_{\mathrm{ev}}(k_x,\kappa)) $ is defined for all $\kappa>0$ and $k_x \in \mathbb R$, including at $\infty$.  We then use $\widehat\psi_\mathrm{in}(k_x,-\kappa)=\widehat\psi_0(k_x,-\kappa)$,  $\widehat\psi_\mathrm{out}(k_x,\kappa)=\widehat\psi_0(k_x,\kappa)$ and
$\widehat\psi_\mathrm{ev}(k_x,\kappa_\mathrm{ev}) = \widehat\psi_\infty(k_x,\kappa_{\mathrm{ev}}(k_x,\kappa))$ in the definition \eqref{eq:def_scattering_state} of the scattering state. Then recall that
\begin{equation}\label{eq:explicit_Psi}
\widehat\Psi_0(k_x,\pm \kappa) = \begin{pmatrix}
	\widehat \psi_0(k_x,\pm \kappa) \\ \pm\ii \kappa \widehat \psi_0(k_x,\pm \kappa)
\end{pmatrix}, \qquad \widehat\Psi_\infty(k_x,\kappa_{\mathrm{ev}}) = \begin{pmatrix}
	\widehat \psi_\infty(k_x,\kappa_\mathrm{ev}) \\ \ii\kappa_\mathrm{ev}\widehat \psi_\infty(k_x,\kappa_\mathrm{ev})
\end{pmatrix}\,,
\end{equation}
which provides a fully explicit expression for the scattering amplitude $S$. In particular, it is regular near infinity.

\subsection{Expansion near $\infty$} 

In order to expand the scattering amplitude near $\infty$, we focus on the expansion of $g$. The latter requires to expand $\widehat\psi_0(k_x,\kappa)$ and $\widehat\psi_\infty(k_x,\kappa_{\mathrm{ev}})$ near $\infty$. Such an expansion is a  bit tricky. Even though the limits are single-valued, the next terms are not, as they are $k_x,k_y$ dependent. Some of these terms are relevant in the expansion of $g$ and have to be properly included. We shall then distinguish in each section the radial part, which depends only on the variable $k = \sqrt{k_x^2+\kappa^2}$ and can be simply expanded, from the directional part which depends on $k_x, \kappa$ and $\kappa_\mathrm{ev}$. Thus we get
\begingroup
\renewcommand*{\arraystretch}{1.5}
$$
\widehat\psi_0(k_x,\kappa) = \begin{pmatrix}
	(k_x-\ii \kappa) \left(\frac{1}{2\epsilon k^2} + o\left(\frac{1}{k^3}\right)\right)\\
	-1+\frac{1}{8\epsilon k^2} + o\left(\frac{1}{k^3}\right)
\end{pmatrix}
$$
For the expansion of $\widehat\psi_\infty$, we first notice that $\tilde k$, $\tilde q$ and $\tilde \omega_\pm$ are all functions of $k^2=k_x^2+\kappa^2$ only, so that they can be expanded as well as a radial part. The remaining terms depends on $k_x$ and $\kappa_{\mathrm{ev}}$. Explicitly we get
$$
\widehat\psi_\infty(k_x,\kappa_\mathrm{ev}) = \begin{pmatrix}
	1+\frac{1}{8\epsilon k^2} + o\left(\frac{1}{k^3}\right)\\
	(k_x+\ii \kappa_\mathrm{ev}) \left(-\frac{1}{2\epsilon k^2} + o\left(\frac{1}{k^3}\right)\right)\\
\end{pmatrix}
$$
\endgroup
Notice that $\kappa_\mathrm{ev} \sim \sqrt{2 k_x^2 + \kappa^2}$ as $k_x,\kappa \to \infty$. This way one has for example
$$
\dfrac{k_x-\ii \kappa}{k^2}\to 0, \qquad \dfrac{k_x+\ii \kappa_\mathrm{ev}}{k^2} \to 0
$$
as $k_x,\kappa \to \infty$, independent from the direction towards infinity, whereas expressions like 

\noindent $k_x(k_x-\ii \kappa)k^{-2}$ have a limit which depends on the $(k_x,\kappa)$ direction. Such expressions naturally appear in the scattering amplitude through $\widehat\Psi_0$ and $\widehat\Psi_\infty$, see \eqref{eq:explicit_Psi}, which involve $\kappa \widehat\psi_0$ and $\kappa_{\mathrm{ev}} \widehat\psi_\infty$, and through $A\widehat\Psi_0$ and $A\widehat\Psi_\infty$, which involve $k_x \widehat\psi_0$ and $k_x \widehat\psi_\infty$. Consequently, an expansion of the radial part of $\widehat\psi_0$ and $\widehat\psi_\infty$ up to order $k^{-3}$ is required to get and expansion of $g$ up to order 1 in the limit $k_x,\kappa \to \infty$. The non radial part of the expansion will be first kept as it is and specified later when passing to dual variables. Thus we will use
\begingroup
\renewcommand*{\arraystretch}{1.7}
\begin{equation}\label{eq:expanded_sections}
\widehat\psi_0(k_x,\kappa) = \begin{pmatrix}
 \dfrac{k_x-\ii \kappa}{2\epsilon k^2} \\
	-1+\dfrac{1}{8\epsilon k^2} 
\end{pmatrix} + o\left(\dfrac{1}{k^2}\right), \qquad \widehat\psi_\infty(k_x,\kappa_\mathrm{ev}) = \begin{pmatrix}
1+\dfrac{1}{8\epsilon k^2} \\
 -\dfrac{k_x+\ii \kappa_\mathrm{ev}}{2\epsilon k^2} \\
\end{pmatrix}+ o\left(\dfrac{1}{k^2}\right)
\end{equation}
\endgroup

\subsection{Detecting anomalies \label{sec:detecting_anomalies}} 

From the explicit expression of $g$ from Proposition~\ref{prop:S_explicit} we compute its expansion near $\infty$ using the expansion of $\widehat\psi_0$ and $\widehat\psi_\infty$ from the previous section and keeping all the terms that do not vanish in the limit:
\begin{equation}
	g(k_x,\kappa) = g_\infty(k_x,k)+o(1),
\end{equation}
with $o(1)$ in the limit $k_x^2+ \kappa^2 \to \infty$, independent from the direction. The exact expression for $g_\infty$ depends on the boundary condition.  The winding number of $S$ near infinity is thus reduced to the study the winding phase of $g_\infty$ computed along the curve $\Gamma_\delta(\lambda)$ with dual variables $\lambda_x \in \mathbb [-\lambda,\lambda]$ and $\delta>0$, see \eqref{eq:defGammadelta}. Then consider
\begin{equation}\label{eq:defg0}
	g_0(\lambda_x,\delta) =  g_\infty\Big(-\dfrac{\lambda_x}{\lambda_x^2+\delta^2},\dfrac{\delta}{\lambda_x^2+\delta^2}\Big).
\end{equation}

\begin{lem}\label{lem:asymptotic_winding}
	Assume the existence of $M>0$ as well as $\lambda_0>0$ and $\delta_0>0$ sufficiently small such that $|g_0|>M$ and $g_0$ is continuously differentiable on $(0,\delta_0]\times [-\lambda_0,\lambda_0]$. Then 
	$$
	\lim_{\lambda \to 0}\lim_{\delta\to 0} \int_{\Gamma_{\delta}(\lambda)} g^{-1}\dd g = \lim_{\lambda \to 0}\lim_{\delta\to 0} \int_{-\lambda}^\lambda g_0^{-1}\partial_{\lambda_x} g_0 \dd \lambda_x
	$$
\end{lem}
\begin{proof} The winding of $g$ reads
	$$
	\int_{\Gamma_\delta(\lambda)} g^{-1}\dd g = \int_{-\lambda}^\lambda \tilde g^{-1} \partial_{\lambda_x} \tilde g \dd \lambda_x, \qquad  \tilde g(\lambda_x,\delta) = g\Big(-\dfrac{\lambda_x}{\lambda_x^2+\delta^2},\dfrac{\delta}{\lambda_x^2+\delta^2}\Big).
	$$
	Since $k_x^2+\kappa^2 = (\lambda_x^2+\delta^2)^{-1}$, one has $\tilde g(\lambda_x,\delta) = g_0(\lambda_x,\delta)+o(1)$ with $o(1)$ in the limit $\lambda_x^2+\delta^2\to 0$, independent from the direction. 
	 Thus we decompose
	$$
	\tilde g=g_0+R=g_0\left(1+\dfrac{R}{g_0}\right).
	$$
	with $R(\lambda_x,\delta)\to 0$ as $\lambda_x^2+\delta^2\to 0$. Since $|g_0|>M$ on $(0,\delta_0]\times [-\lambda,\lambda]$ then $\tfrac{R}{G}\to 0$ as $\lambda_x^2+\delta^2\to 0$.
	
	Moreover, the multiplicative property of winding phase gives
	$$
 \int_{-\lambda}^\lambda \tilde g^{-1} \partial_{\lambda_x} \tilde g \dd \lambda_x = \int_{-\lambda}^\lambda g_0^{-1} \partial_{\lambda_x} g_0\dd \lambda_x + \int_{-\lambda}^\lambda \left(1+\dfrac{R}{g_0}\right)^{-1}  \partial_{\lambda_x}\left(1+\dfrac{R}{g_0}\right)\dd \lambda_x.
	$$
	Let $\varepsilon>0$ and $\delta_1,\lambda_1$  such that $\left|\tfrac{R}{g_0}\right|<\varepsilon$ on $(0,\delta_1]\times [-\lambda_1,\lambda_1]$.
	
	Consequently,  for $|\lambda|< \min(\lambda_0,\lambda_1)$ and $\delta < \min(\delta_0,\delta_1)$ one has
	$$ 
	\left|\arg\left(1+\dfrac{R}{g_0}\right)\right| \leq \dfrac{\varepsilon}{1-\varepsilon}
	$$
	so that
	$$
	\left|\int_{-\lambda}^\lambda \left(1+\dfrac{R}{g_0}\right)^{-1}  \partial_{\lambda_x} \left(1+\dfrac{R}{g_0}\right)\dd \lambda_x\right|  \leq \dfrac{2\varepsilon}{1-\varepsilon}
	$$
	Thus the winding phase of $\left(1+\tfrac{R}{g_0}\right) $ goes to $0$ as $\delta\to 0$ and $\lambda \to 0$.
\end{proof}

One can check that the assumption of the lemma above is satisfied for all explicit $g_0$ below by checking that  the limits of $|g_0|$ as $\delta\to 0$ or $\lambda \to 0$ are either finite and non vanishing or infinite. Thus, to compute $w_\infty$ we can simply replace $S$ by $S_0$ with
$$
S_0(\lambda_x,\delta) = \dfrac{g_0(\lambda_x,-\delta)}{g_0(\lambda_x,\delta)}.
$$

The next step is to extract the dominant scale from $g_0$ in the limit $\delta\to 0$ and $\lambda \to 0$. First, due to the expression of $S_0$ we can replace $g_0$ by any function that is even in $\delta$, typically $g_0 \mapsto g_0 f$ with $f(\delta,\lambda) = (\lambda_x^2+\delta^2)^2$ or $(\lambda_x^2+\delta^2)^2$, see other examples below. We shall keep the same symbol for the simplified expression of $g_0$.

Finally, we use the following proposition to extract the leading contribution from $g_0$.
\begin{prop}\label{lem:dominant_scale}
Let $g_0 : \mathbb R \times \mathbb R_+^* \to \mathbb C^*$ continuously differentiable. Let $s >0$,  $u=\frac{\lambda_x}{\delta^s}$ and $r>0$. Decompose 
$$
\delta^{-r} g_0(u \delta^s,\delta) = G(u) + R(u,\delta)
$$
with $G : \mathbb R \to  \mathbb C^*$ and $R:  \mathbb R \times \mathbb R_+^* \to \mathbb C$ continuously differentiable and such that 
\begin{enumerate}
	\item For any $u \in \mathbb R$, $\displaystyle \lim_{\delta \to 0} R(u,\delta) =0$,
	\item For $\delta >0$, $$\displaystyle \dfrac{R(u,\delta)}{G(u)} \sim C u^\gamma \delta^{s \gamma}$$ as $u \to +\infty$, with $C \in \mathbb C^*$ independent from $\delta$, and $\gamma >0$.
	\item Assumption 2 also holds as $u \to -\infty$ with possibly distinct constants $\widetilde C$ and $\widetilde \gamma$.
\end{enumerate}
Then, 
\begin{equation}\label{eq:lem_dominant_scale}
	\lim_{\lambda \to 0}\lim_{\delta \to 0} \int_{-\lambda}^\lambda g_0^{-1} \partial_{\lambda_x} g_0 \dd \lambda_x = \int_{-\infty}^{\infty} G^{-1}(u)G'(u) \dd u.
\end{equation}
\end{prop}
\begin{proof}
	First notice that for $\delta>0$ the winding phase of $g_0$ and $\delta^{-r} g_0$ are the same. We change the integration variable $\lambda_x$ to $u=\tfrac{\lambda_x}{\delta^s}$ and replace $g_0$ by $g_0\delta^{-r} = G+R$. Then we rewrite $G+R=G(1+\tfrac{R}{G})$ and use the additive property of winding phase:
	\begin{align}\label{eq:proof_lem_dominant_scale}
		\int_{-\lambda}^\lambda g_0^{-1} \partial_{\lambda_x} g_0 \dd \lambda_x &= \int_{-\frac{\lambda}{\delta^s}}^{\frac{\lambda}{\delta^s}} g_0^{-1} \partial_{u} g_0 \dd u\cr
		& = \int_{-\frac{\lambda}{\delta^s}}^{\frac{\lambda}{\delta^s}} G^{-1}(u)  G'(u) \dd u + \int_{-\frac{\lambda}{\delta^s}}^{\frac{\lambda}{\delta^s}} \Big(1+\frac R G\Big)^{-1}(u,\delta)  \partial_u \Big(1+\frac R G\Big)(u, \delta) \dd u.
	\end{align}
	The integrand of the first term does not depend on $\delta$, so in the limit $\delta \to 0$ we get the right hand side of \eqref{eq:lem_dominant_scale}. 
	
	Then we claim that for all $\varepsilon>0$ there exist $\delta_0 >0$ and $\lambda_0>0$ such that 
	$$
	\forall u \in \left[-\frac{\lambda}{\delta^s}, \frac{\lambda}{\delta^s}\right], \quad \forall \delta < \delta_0, \quad \forall \lambda < \lambda_0, \qquad \left|\dfrac{R(u,\delta)}{G(u)}\right| \leq \varepsilon\,,
	$$
	which implies, similarly to the proof of the Lemma above, that the winding phase of $u \mapsto 1+\tfrac{R}{G}$ is arbitrarily small for $u \in \left[-\tfrac{\lambda}{\delta^s}, \tfrac{\lambda}{\delta^s}\right]$ if  $\delta$, $\lambda$ are sufficiently small. Consequently, the second term in the right hand side of \eqref{eq:proof_lem_dominant_scale} vanishes in the limit $\delta \to 0$, $\lambda\to 0$, and we get \eqref{eq:lem_dominant_scale}.

	 Let $\varepsilon>0$. In order to prove the claim above, we split the study of the ratio by into the asymptotic and finite regions. First, by Assumption 2, there exists a $u_1>0$ such that, 
	$$
	\forall u>u_1, \qquad  \left|\dfrac{R(u,\delta)}{G(u)}\right| \leq (1+\varepsilon) C u^\gamma \delta^{s \gamma}
	$$
	Let $\delta_1>0$ such that $\tfrac{\lambda}{\delta_1^\alpha}>u_1$. For $\delta <\delta_1$ and $u$ such that $u_1 < u \leq  \tfrac{\lambda}{\delta^s}$ one has
	$$
	 \left|\dfrac{R(u,\delta)}{G(u)}\right| \leq (1+\varepsilon) C \lambda^\gamma,
	$$
	so that for $\lambda<\lambda_1$ sufficiently small the right hand side is smaller than $\varepsilon$. 
	
	Similarly, at $-\infty$ there exists $u_2>0$ and $\delta_2>0$ such that for $\delta <\delta_2$ and $u$ such that $- \tfrac{\lambda}{\delta^s}\leq u < -u_2$ one has
	$$
	\left|\dfrac{R(u,\delta)}{G(u)}\right| \leq (1+\varepsilon) \widetilde C \lambda^{\widetilde \gamma}\,.
	$$
	So that for $\lambda<\lambda_2$ sufficiently small the right hand side is smaller than $\varepsilon$. 
	
	Finally, $G$ and $R$ are continuous on $[-u_2,u_1]$ and $G(u)\neq 0$. Let $M >0$ such that $|G(u)|>M$ for $u \in [-u_2,u_1]$. Moreover, by Assumption 1 one has $R(u,\delta)\to 0$ as $\delta\to 0$. For $u \in [-u_2,u_1]$, this limit is uniform in $\delta$. So there exists a $\delta_3>0$ such that 
	$$
	\forall 0 <\delta < \delta_3, \forall  u \in [-u_2,u_1], \qquad |R(u,\delta)|< \dfrac{\varepsilon}{M}\,,
	$$
	which implies
	$$
	\forall 0< \delta < \delta_3, \forall  u \in [-u_2,u_1], \qquad \left|\dfrac{R(u,\delta)}{G(u)}\right|< \varepsilon\,.
	$$
	Putting all together, we get the claim if $\delta < \min(\delta_1,\delta_2,\delta_3)$ and $\lambda<\min(\lambda_1,\lambda_2)$.
\end{proof}

This lemma drastically simplifies the computation of the winding number since the leading contribution $G$ usually consists of only a few terms of the numerous ones of $g_0$.
\begin{cor}\label{cor:dominant_scale}
	Assume 
	\begin{align*}
	\delta^{-r} g_0(u \delta^s,\delta) = G_+(u) + R_+(u,\delta), \qquad \delta^{-r} g_0(u \delta^s,-\delta) = G_-(u) + R_-(u,\delta)
	\end{align*}
	 with $G_\pm,R_\pm$ satisfying the assumption of Proposition~\ref{lem:dominant_scale}. Then
	$$
	w_\infty = \int_{-\infty}^{\infty}  \mathcal S^{-1}(u) \mathcal S'(u) \dd u, \qquad \mathcal S(u) = \dfrac{G_-(u)}{G_+(u)}\,.
	$$
\end{cor}

A detailed application of this result with a specific boundary condition is given in Section~\ref{sec:casestudy} below. Before to move on with the full classification, we need to anticipate one possible issue: It may occasionally happen for specific boundary conditions below that $G_+(u_0) = 0 = G_-(u_0)$ for some $u_0 \in \mathbb R$. In that case Proposition~\ref{lem:dominant_scale} does not apply for two reasons: 
	\begin{enumerate}
		\item The winding phases of $G_+$ and $G_-$ are ill-defined  
	  \item $\tfrac{R_+}{G_+}$  and $\tfrac{R_-}{G_-}$  may be very large near $u_0$.
	  	\end{enumerate}
	  However, if $G_+$ and $G_-$ are such that the ratio $\mathcal S= \tfrac{G_-}{G_+}$ has a well defined the winding phase then point 1 is actually not an issue to compute $w_\infty$. Thus we only need to check that $R_-$ and $R_+$ do not contribute to the winding phase of $S$ in that case. The proposition below provides another criterion and generalizes Proposition~\ref{lem:dominant_scale}.
	

\begin{prop}\label{lem:dominant_scale2}
	Let $g_0 : \mathbb R \times \mathbb R_+^* \to \mathbb C^*$ continuously differentiable. Let $s >0$, $u=\tfrac{\lambda_x}{\delta^s}$ and $r>0$. Decompose 
	$$
	\delta^{-r} g_0(u \delta^s,\delta) = G_+(u) + R_+(u,\delta), \qquad \delta^{-r} g_0(u \delta^s,-\delta) = G_-(u)+R_-(u,\delta)
	$$
	with $G_+,G_- : \mathbb R \to  \mathbb C$ and $R_+, R_-:  \mathbb R \times \mathbb R_+^* \to \mathbb C$ continuously differentiables and such that Assumptions 1, 2 and 3 of Proposition~\ref{lem:dominant_scale} hold. Assume moreover that $G_+(u_0) = 0 = G_-(u_0)$ for some $u_0 \in \mathbb R$, but such that the ratios 
	$$
	\mathcal S(u) = \dfrac{G_-(u)}{G_+(u)}, \qquad F(u,\delta)=\dfrac{1+\frac{R_-(u,\delta)}{G_-(u)}}{1+\frac{R_+(u,\delta)}{G_+(u)}}
	$$
	remain continuously differentiables near $u_0$ and that 
	$$
	\lim_{\delta \to 0} \lim_{u \to u_0} F(u,\delta) = 1.
	$$
	Then,
	\begin{equation}
	\lim_{\lambda \to 0} 	\lim_{\delta \to 0} \int_{-\lambda}^\lambda S_0^{-1} \partial_{\lambda_x} S_0 \dd \lambda_x = \int_{-\infty}^{\infty} \mathcal S^{-1}(u) \mathcal S'(u)\dd u.
	\end{equation}
\end{prop}

\begin{proof}
	We proceed as in the proof of Proposition~\ref{lem:dominant_scale}: multiply numerator and denominator by $\delta^{-r}$, change the variable and split the integral:
	\begin{align}
		\int_{-\lambda}^\lambda S_0^{-1} \partial_{\lambda_x} S_0 \dd \lambda_x &= \int_{-\frac{\lambda}{\delta^s}}^{\frac{\lambda}{\delta^s}} S_0^{-1} \partial_{u} S_0 \dd u\cr
		& = \int_{-\frac{\lambda}{\delta^s}}^{\frac{\lambda}{\delta^s}}S^{-1}(u) S'(u) \dd u + \int_{-\frac{\lambda}{\delta^s}}^{\frac{\lambda}{\delta^s}} F^{-1}(u,\delta)  \partial_u F(u, \delta) \dd u.
	\end{align}
    The first term gives the expected result and we show that the second has no winding phase in the limit. We split as before between small and large $u$
    \begin{align}
    	&\int_{-\frac{\lambda}{\delta^s}}^{\frac{\lambda}{\delta^s}} F^{-1}(u,\delta)  \partial_u F(u, \delta) \dd u \cr = & 
    	\int_{M}^{\frac{\lambda}{\delta^s}} \Big(1+\frac{R_+}{G_+} \Big)^{-1} \partial_u \Big(1+\frac{R_+}{G_+}\Big) \dd u + \int_{M}^{\frac{\lambda}{\delta^s}} \Big(1+\frac{R_-}{G_-} \Big)^{-1} \partial_u \Big(1+\frac{R_-}{G_-}\Big) \dd u \cr
    	& + \int_{-\frac{\lambda}{\delta^s}}^{M} \Big(1+\frac{R_+}{G_+} \Big)^{-1} \partial_u \Big(1+\frac{R_+}{G_+}\Big) \dd u + \int_{-\frac{\lambda}{\delta^s}}^{M} \Big(1+\frac{R_-}{G_-} \Big)^{-1} \partial_u \Big(1+\frac{R_-}{G_-}\Big) \dd u \cr
    	& + \int_{-M}^M F^{-1}(u,\delta)  \partial_u F(u, \delta).
    \end{align}
	We treat the four first terms as before: the asymptotic properties are the same, so similarly to the proof of Proposition~\ref{lem:dominant_scale}, we show that for any $\varepsilon> 0$ there exist $\delta$ and $\lambda$ sufficiently small and some $M>|u_0|$ such that
	$$
	|u|>M \quad \Rightarrow \quad \left| \dfrac{R_\pm(u,\delta)}{G_\pm(u)}\right| \leq \varepsilon 
	$$
	so that in these regions the associated winding phase of $1+\tfrac{R_\pm}{G_\pm}$ vanish as $\delta \to 0$ and $\lambda \to 0$.
	
	Finally, for the fifth term, we use the fact that $F$ is continuous for all $u \in [-M,M]$ and that  $F(u,\delta) \to 1$ and $\delta \to 0$. Moreover, the limit can be taken uniformly on $u\in [-M,M]$ so that $|F(u,\delta)-1|\leq \epsilon$ for all $u\in [-M,M]$ and $\delta, \lambda$ sufficiently small. Again, $F$ has no winding phase in that region in the limit.
\end{proof}

\subsection{Case study \label{sec:casestudy}}

\begin{exmp}\label{exmp:case_study}
	We illustrate Lemma~\ref{lem:asymptotic_winding}, Proposition~\ref{lem:dominant_scale} and Corollary~\ref{cor:dominant_scale} above in the following simple case of the classification where the boundary condition is
	\begin{equation}\label{eq:Acasestudy}
		A_0 = \begin{pmatrix}
			1 & 0 & 0 & 0\\
			0 & \alpha & 0 & 1
		\end{pmatrix}, \qquad A_1= \begin{pmatrix}
			0 & 0 & 0 & 0\\
			b_{21} & \ii \beta & 0 & 0
		\end{pmatrix} := \begin{pmatrix}
			B & 0 
		\end{pmatrix},
	\end{equation}
	with $b_{12}\in \mathbb C$ and $\alpha, \beta \in \mathbb R$, which is part of class $\mathfrak{A}_{1,4}$. We further assume $\beta \neq 0$ here. Notice that for $b_{12}=\alpha=0$ and $\beta=-1$ we recover condition (b) from Examples~\ref{example} and \ref{example_continuated}. Inserting \eqref{eq:expanded_sections} and \eqref{eq:Acasestudy} into \eqref{eq:S_explicit} we get
	\begin{align}
		g(k_x,\kappa) = \frac{1}{64 k^4 \epsilon ^4}\Big(& \,-\alpha -\ii \kappa +\beta  k_x +64 k^4 \epsilon ^4 \left(\alpha +\ii \kappa -\beta  k_x\right)\cr &+16 \epsilon ^2 \left(\kappa +\ii k_x\right) \left(\kappa _{\text{ev}}-\ii k_x\right) \left(\alpha +\ii \kappa _{\text{ev}}-\beta  k_x\right)\Big) + o(1)
	\end{align}
	Many terms vanish in the limit $k\to \infty$ so this further simplifies to $g(k_x,\kappa) = g_\infty(k_x,\kappa) +o(1)$ with
	$$
	g_\infty(k_x,\kappa) = \alpha + \ii \kappa - \ii k_x \beta.
	$$
	Passing to dual variables we get
	$$
	g_0(\lambda_x,\delta) = g_\infty\Big(-\dfrac{\lambda_x}{\lambda_x^2+\delta^2},\dfrac{\delta}{\lambda_x^2+\delta^2}\Big) =\frac{ \alpha  \lambda_x^2+\beta \lambda_x +\alpha  \delta ^2+\ii \delta }{\delta ^2+\lambda_x^2},
	$$
	and one can check that Lemma~\ref{lem:asymptotic_winding} applies. Then as discussed above, the denominator does not contribute to the winding phase so we replace $g_0$ by $(\delta ^2+\lambda_x^2)g_0$ (and keep the same symbol for $g_0$), leading to
	\begin{equation}
		g_0(\lambda_x,\delta) = \alpha  \lambda_x^2+\beta \lambda_x +\alpha  \delta ^2+\ii \delta
	\end{equation}
	
	Now, the dominant scale is given by $s=r=1$, namely
	$$
	\delta^{-1}g(u\delta,\pm\delta) = G_\pm(u) + R_\pm(u,\delta) 
	$$
	with
	\begin{equation}
		G_+(u) = \ii +\beta u, \qquad G_-(u)=-\ii + \beta u.\qquad R_+(u) = R_-(u) = \alpha \delta (u^2+1)  
	\end{equation}
	One can check that Proposition~\ref{lem:dominant_scale} applies. In particular, since we assume $\beta \neq 0$
	$$
	\dfrac{R_\pm}{G_\pm} \mathop{\sim}_{u \to +\infty} \dfrac{\alpha}{\beta} u \delta,
	$$
	namely $\gamma =1$, and similarly as $u \to -\infty$. Thus we are left with the winding number of 
	$$
	\mathcal S = \dfrac{G_-(u)}{G_+(u)} = \dfrac{-\ii + \beta u}{\ii + \beta u}
	$$
	which can be computed as twice the winding phase of $G_-(u)$ since $G_+(u)= G^*_-(u)$. The real part of $G_-$ is constant negative so that
	the winding phase is simply computed by the limits of $\arg(G_-(u))$ as $u \to \pm \infty$. For $\beta >0$, the winding phase of $G_-$ is $\pi$ and for $\beta<0$ it is $-\pi$. Thus in that example we infer $w_\infty=\sign(\beta)$, regardless of the values of $\alpha$ and $b_{21}$. The relevant quantities of this example are illustrated in Figure~\ref{fig:exmp} below.
	
	\begin{figure}[htb]
		\includegraphics[scale=0.7]{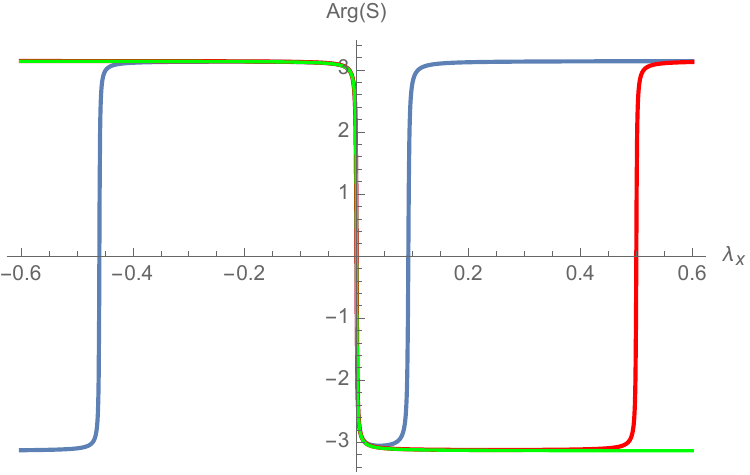}
		\hspace{0.5cm}
		\raisebox{0.5cm}{\includegraphics{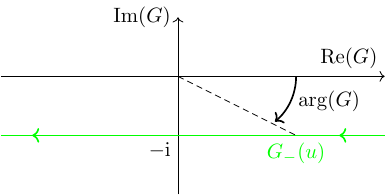}}
		\caption{Argument of the scattering amplitude near $\lambda_x=0$ for $\alpha=2$, $\beta=-1$ and $b_{12}=0$. Moreover, $m=1, \epsilon =0.1$ and $\delta=0.001$. (Left) Argument of $S$ with respect to $\lambda_x$. The blue curve corresponds to the exact $S$, which winds at $\lambda_x=0$ but also at near $-0.5$ and $0.1$, detecting edge modes at finite $k_x \approx 2$ and $-10$ (compare with Figure~\ref{fig:edgemodes}, middle). The red curve corresponds to $S_0$, computed with $g_0$, which approximates $g$ well near $0$ but also keeps track (poorly) of unwanted finite edge modes. The green curve corresponds to $\mathcal S$ computed with $G_\pm$ (actually $\delta G_\pm(\lambda_x/\delta)$ here, to compare with other curves), which exactly captures only the winding of $S$ near infinity. (Right) Complex curve $G_-(u) = -\ii + \beta u$ for $\beta<0$. Its argument changes by $-\pi$ as $u$ goes from $-\infty$ to $+\infty$. \label{fig:exmp}}
	\end{figure}
\end{exmp}

\section{Anomaly classification \label{sec:anomaly_classification}}

\paragraph{General strategy.} We apply the statements from Section~\ref{sec:detecting_anomalies} to each class of self-adjoint  boundary conditions from Table~\ref{tab:self-adjoint_classes}, going from $g$ to its leading contribution $G$, and inferring $w_\infty$. More precisely:
\begin{enumerate}
	\item Pick a matrix $A=A_0+\ii k_x A_1$ from Table~\ref{tab:self-adjoint_classes} and insert it in expression \eqref{eq:S_explicit} for $g$, together with the asymptotic expansion of the bulk eigensections \eqref{eq:expanded_sections}. This leads to $g(k_x,\kappa) = g_\infty(k_x,k)+o(1)$. We shall not detail this step below and give only the final expression. 
	\item Pass from $(k_x,\kappa)$ to dual variables via \eqref{eq:defg0}, compute $g_0(\lambda_x,\delta)$ and simplify it.
	\item Find the dominant scale $\delta^{-r} g_0(u \delta^s, \pm \delta) = G_\pm(u)  + R_\pm(u)$ and check that $G_\pm$ and $R_\pm$ satisfy either the assumptions of Proposition~\ref{lem:dominant_scale} or that of Proposition~\ref{lem:dominant_scale2}. We shall not always give the explicit expressions of $R_\pm$ but they can be inferred from $g_0$ and $G_\pm$. 
	\item Study the complex curves $G_\pm$ and compute their winding phases, from which we deduce $w_\infty$.
	\item For specific values of the parameters in $A$, Assumption 2 of Proposition~\ref{lem:dominant_scale} may fail to hold so that another scale has to be found instead. This leads to various subcases for which we repeat steps 4 and 5.
\end{enumerate}

\subsection{Class $\mathfrak{A}_{1,2}$} In that case we have
$$
A_0 = \begin{pmatrix}
	1 & 0 & 0 & 0\\
	0 & 1 & 0 & 0
\end{pmatrix}, \qquad A_1= \begin{pmatrix}
	b_{11} & b_{12} & 0 & 0\\
	b_{21} & b_{22} & 0 & 0
\end{pmatrix} := \begin{pmatrix}
B & 0 
\end{pmatrix},
$$
where $B \in \mathrm{M}_{2}(\mathbb C)$ is arbitrary. The expansion of $g$ near $\infty$ leads to
$
g(k_x,\kappa) = g_\infty(k_x,k)+o(1)
$
where $o(1)$ in the limit $k_x^2+ \kappa^2 \to \infty$, independent from the direction, and with 
\begin{equation}
g_\infty(k_x,\kappa) = -\det(B) k_x^2 + \ii \tr(B) k_x  + \det(B) \dfrac{k_x^2\big(k_x^2+\kappa \kappa_{\mathrm{ev}} + \ii k_x ( \kappa_{\mathrm{ev}} - \kappa )\big)}{4\epsilon^2 k^4}
+1.
\end{equation}
We then move to the dual variables $\lambda_x,\delta$ via \eqref{eq:defg0}, leading to
\begin{align}\label{eq:g0_Dirichlet}
	g_0(\lambda_x,\delta)= & \, \frac{1}{4 \epsilon ^2 \left(\delta ^2+\lambda _x^2\right){}^2} \Big(4 \delta ^4 \epsilon ^2+\ii \delta  \det(B) \lambda _x^2 \sqrt{\delta ^2+2 \lambda _x^2}+\det(B) \lambda _x^3 \sqrt{\delta ^2+2 \lambda _x^2}\cr &+\ii \delta  \det(B) \lambda _x^3+\det(B) \lambda _x^4-4 \det(B) \epsilon ^2 \lambda _x^2-4 \ii \delta ^2 \tr(B) \epsilon ^2 \lambda _x\cr &-4 \ii \tr(B) \epsilon ^2 \lambda _x^3+8 \delta ^2 \epsilon ^2 \lambda _x^2+4 \epsilon ^2 \lambda _x^4\Big)
\end{align}


\paragraph{Case $\det B=0$.} In that case one has $g(\lambda_x, \delta) = g_0(\lambda_x, -\delta)$ so that $S_0=1$ and $w_\infty=0$. 

\paragraph{Case $\det B\neq 0$.} The denominator does not contribute to $w_\infty$ so we replace $g_0$ by $(\delta ^2+\lambda _x^2)^2g_0$, and we 
we extract the dominant scale by considering
$$
	\delta^{-4} g_0(u \delta^2,\delta) = G_+(u)+R_+(u,\delta)
$$
with
\begin{align}
	G_+(u) = & \, 1- \ii \tr(B) u - \det(B) u^2,\\
	R_+(u,\delta) = &\, \frac{\delta ^4 \det(B) u^4}{4 \epsilon ^2}+\frac{\ii \delta ^3 \det(B) u^3}{4 \epsilon ^2}+\frac{\delta ^2 \det(B) u^3 \sqrt{\delta ^2+2 \delta ^4 u^2}}{4 \epsilon ^2}\cr &+\frac{\ii \delta  \det(B) u^2 \sqrt{\delta ^2+2 \delta ^4 u^2}}{4 \epsilon ^2}-\ii \delta ^2 \tr(B) u^3+\delta ^4 u^4+2 \delta ^2 u^2.
\end{align}
One has $R_+(u,0)=0$ and 
$$
\dfrac{R(u,\delta)}{G(u)} \sim -\dfrac{(1\pm\sqrt{2} )\det(B)+4 \epsilon ^2}{4 \det(B) \epsilon ^2}u^2 \delta ^4 
$$
as $u \to \pm \infty$, so that Proposition~\ref{lem:dominant_scale} applies\footnote{One can check that Proposition~\ref{lem:dominant_scale} also applies in the fine-tuned case where $\det(B)=-\tfrac{1\pm \sqrt{2}}{4\epsilon^2}$, then a lower order computation leads to $\dfrac{R(u,\delta)}{G(u)} \sim -\frac{\ii (\sqrt{2}+1) \delta ^2 u (\delta +\tr (B))}{4 \epsilon ^2}$} when $G_+$ has no zero.

Similarly, replacing $\delta$ by $-\delta$ in \eqref{eq:g0_Dirichlet} one gets $\delta^{-4}g_0(u \delta^2,-\delta) = G_-(u)+R_-(u,\delta)$ with
\begin{align}
	G_-(u) = &\, G_+(u),\\
	R_-(u,\delta) = &\, \frac{\delta ^4 \det(B) u^4}{4 \epsilon ^2}-\frac{\ii \delta ^3 \det(B) u^3}{4 \epsilon ^2}+\frac{\delta ^2 \det(B) u^3 \sqrt{\delta ^2+2 \delta ^4 u^2}}{4 \epsilon ^2}\cr &-\frac{\ii \delta  \det(B) u^2 \sqrt{\delta ^2+2 \delta ^4 u^2}}{4 \epsilon ^2}-\ii \delta ^2 \tr(B) u^3+\delta ^4 u^4+2 \delta ^2 u^2.
\end{align}
and one can check that Proposition~\ref{lem:dominant_scale} also applies.
In particular
$$
\mathcal S(u) = \dfrac{G_-(u)}{G_+(u)}=1
$$
so that $w_\infty=0$.

\begin{rem}[About vanishing cases] $G_+$ and $G_-$ may occasionally vanish at some $u_0 \in \mathbb R$ for exceptional values of $B$, but $w_\infty$ remains the same. Let $\det(B)=d_r + \ii d_i$ and $\tr(B)=t_r+\ii t_i$ with $d_r,d_i,t_r$ and $t_i \in \mathbb R$. 
One has 
$$\Re(G_+(u)) = 1+t_i u -d_r u^2, \qquad  \Im(G_+(u)) = -u(t_r+u d_i).$$
The function $G_+$ vanishes in the following three cases :
\begin{enumerate}
	\item If $d_r=0, t_i \neq 0, d_i \neq 0$ and $d_i=-t_rt_i$ at $u_0=-\tfrac{1}{t_i}$
	\item If $d_r\neq 0$, $\Delta = t_i^2+4d_r >0$ and
	\begin{enumerate}
		\item  $t_r=d_i=0$ at $u_{0\pm }= \tfrac{t_i \pm \sqrt{\Delta}}{2d_r}$,
		\item  $d_i\neq 0$  with $2t_rd_r = - d_i(t_i + \sqrt \Delta)$ or $2t_rd_r = - d_i(t_i - \sqrt \Delta)$, at $u_0 = -\tfrac{t_r}{d_i}$.
	\end{enumerate}
\end{enumerate}

In all cases above one could check that Proposition~\ref{lem:dominant_scale2} applies. However, one still has $G_+=G_-$, so we can instead replace $g_0$ by $g_0 + z_0$ for some $z_0 \in \mathbb C^*$ with $|z_0|$ sufficiently small so that it does not change the winding phase of $g_0$. The computation above leads to 
$$g_0(\lambda_x,\delta) + z_0 = G_+(u)+ z_0 +R_+(u,\delta), \qquad g_0(\lambda_x,-\delta) + z_0 = G_-(u) + z_0 +R_-(u,\delta)$$ with again $G_-(u)=G_+(u)$. Moreover, $z_0$ can be chosen such that $G_+(u)+z_0 \in \mathbb C^*$ for all $u \in \mathbb R$ and, consequently, $G_-(u)+z_0 \in \mathbb C^*$ . The asymptotic properties are preserved as $u \to \pm \infty$ so that Proposition~\ref{lem:dominant_scale} now applies to $(G_++z_0,\, R_+)$ and $(G_-  +z_0,\, R_-)$. Thus $\mathcal S= \tfrac{G_- + z_0}{G_+ +z_0}=1$ and we get  $w_\infty=0$ also in that case.
\end{rem}

 \paragraph{Summarizing.} Thus $S$ does not wind and $w_\infty =0$ for any boundary condition in this class. 
 
\subsection{Class $\mathfrak{A}_{1,4}$} In this case we have

$$
A_0 = \begin{pmatrix}
	1 & 0 & 0 & 0\\
	0 & \alpha & 0 & 1
\end{pmatrix}, \qquad A_1= \begin{pmatrix}
	b_{11} & 0 & 0 & 0\\
	b_{21} & \ii \beta & 0 & 0
\end{pmatrix} := \begin{pmatrix}
	B & 0 
\end{pmatrix},
$$
with $b_{11}, b_{21} \in \mathbb C$ and $\alpha, \beta \in \mathbb R$. We get $g(k_x,\kappa) = g_\infty(k_x,\kappa)+o(1)$ with
\begin{align}\label{eq:ginfA14}
g_\infty(k_x,\kappa) = & -\det(B)k_x^2-b_{11} k_x \kappa +\ii \kappa +\ii k_x \left(\alpha b_{11}+\ii \beta\right)+\alpha  \cr &
 + \det(B) \dfrac{k_x^2\big(k_x^2+\kappa \kappa_{\mathrm{ev}}+ \ii k_x ( \kappa_{\mathrm{ev}} - \kappa )\big)}{4\epsilon^2 k^4}  +b_{11} \frac{ k_x \kappa_\mathrm{ev}\big(k_x^2+\kappa \kappa_{\mathrm{ev}} + \ii k_x ( \kappa_{\mathrm{ev}} - \kappa )\big)}{4 k^4 \epsilon ^2}\,. 
\end{align}
Passing to dual variables via \eqref{eq:defg0} we get
\begin{align}\label{eq:g0A14}
		g_0(\lambda_x,\delta)= &\, \dfrac{1}{4\epsilon^2(\lambda_x^2+\delta^2)^2}\Big(4 \alpha  \delta ^4 \epsilon ^2-4 \ii \alpha  b_{11} \delta ^2 \epsilon ^2 \lambda _x-4 \ii \alpha  b_{11} \epsilon ^2 \lambda _x^3+b_{11} \delta ^3 \lambda _x-\ii b_{11} \delta ^2 \lambda _x^2\cr &+b_{11} \delta  \lambda _x^2 \sqrt{\delta ^2+2 \lambda _x^2}-\ii b_{11} \lambda _x^3 \sqrt{\delta ^2+2 \lambda _x^2}+2 b_{11} \delta  \lambda _x^3+4 b_{11} \delta  \epsilon ^2 \lambda _x-2 \ii b_{11} \lambda _x^4\cr &+4 \ii \delta ^3 \epsilon ^2+\ii \delta  \det(B) \lambda _x^2 \sqrt{\delta ^2+2 \lambda _x^2}+\det(B) \lambda _x^3 \sqrt{\delta ^2+2 \lambda _x^2}+\ii \delta  \det(B) \lambda _x^3+\det(B) \lambda _x^4\cr &-4 \det(B) \epsilon ^2 \lambda _x^2+8 \alpha  \delta ^2 \epsilon ^2 \lambda _x^2+4 \alpha  \epsilon ^2 \lambda _x^4+4 \beta  \delta ^2 \epsilon ^2 \lambda _x+4 \beta  \epsilon ^2 \lambda _x^3+4 \ii \delta  \epsilon ^2 \lambda _x^2\Big)\,.
\end{align}
The denominator does not contribute to $w_\infty$ so we replace $g_0$ by $(\delta ^2+\lambda _x^2)^2g_0$. Then the dominant scale of $g_0$ depends on some parameter vanishing or not. 

\paragraph{Case $\det(B)\neq0$.} Since  $\det(B)=\ii \beta b_{11}$ one has $b_{11}\neq0$ and $\beta \neq 0$. The dominant scale is 
$$ 
\delta^{-2} g_0(u \delta,\pm \delta) = G_\pm(u) + R_\pm(u)
$$
with 
\begin{equation}
	G_+(u) = u b_{11}(1-\ii \beta u), \qquad G_-(u) = -ub_{11}(1+\ii \beta u)\,.
\end{equation}
$G_+$ and $G_-$ vanish both at $u=0$ but one can check that Proposition~\ref{lem:dominant_scale2} applies, in particular that the ratio $F(u,\delta) \to 1$ as $u \to 0$ and $\delta \to 0$. Thus we are left with the winding number of 
\begin{equation}
	\mathcal S (u)=\dfrac{G_-(u)}{G_+(u)} = -\dfrac{1+\ii \beta u}{1-\ii \beta u} = \dfrac{\ii - \beta u}{\ii + \beta u}
\end{equation}
Numerator and denominator are conjugated so $w_\infty$ is twice the winding phase of the numerator, which has constant positive imaginary part. Its winding phase is the difference of limit of argument between $+\infty$ and $-\infty$. If $\beta>0$ this winding phase is $\pi$. If $\beta <0$ it is $-\pi$. Thus we infer $w_\infty=\sign(\beta)$.

\paragraph{Case $b_{11}=0$ and $\beta \neq 0$.} This case has been completely treated already in Example~\ref{exmp:case_study}, and also leads to $w_\infty=\sign(\beta)$.

\paragraph{Case $b_{11}\neq 0$ and $\beta = 0$.} In that case $\det(B)=0$ and we go back to expression \eqref{eq:g0A14} for $g_0$ and replace it again by $(\delta ^2+\lambda _x^2)^2g_0$. The dominant scale is 
$$
\delta^{-3/2} g_0(u \delta^{1/2},\pm \delta) = G_\pm(u) + R_\pm(u)
$$
%
%
with
\begin{equation}
	G_+(u) = b_{11} u \left(1-\ii \alpha  u^2\right), \qquad G_-(u) = b_{11} u \left(-1-\ii \alpha  u^2\right) = - G^*_+(u)
\end{equation}
$G_+$ and $G_-$ vanish both at $u=0$ but one can check that Proposition~\ref{lem:dominant_scale2} applies, in particular that the ratio $F(u,\delta) \to 1$ as $u \to 0$ and $\delta \to 0$. Thus we are left with the winding number of
$$
\mathcal S(u) = \dfrac{G_-(u)}{G_+(u)} = \dfrac{1-\ii \alpha  u^2}{1+\ii \alpha  u^2}.
$$
Notice that $\Re(1-\ii \alpha  u^2)=1 >0$ so that its argument goes continuously from $\pi/2$ to $\pi/2$ as $u$ goes from $-\infty$ to $\infty$. Consequently, its winding phase vanishes, and similarly for $1+\ii \alpha  u^2$. Thus $w_\infty=0$.

\paragraph{Case $b_{11}=0$ and $\beta = 0$}
In that case we go back to \eqref{eq:ginfA14} which simplifies to $g_\infty(k_x,\kappa) = \alpha + \ii \kappa$ leading to 
$$
g_0(\lambda_x,\delta) = \frac{ \alpha  \lambda_x^2+\alpha  \delta ^2+\ii \delta }{\delta ^2+\lambda_x^2},
$$
Replacing $g_0$ by $g_0(\delta ^2+\lambda_x^2)$, the dominant scale is 
$$
\delta^{-1} g(u \delta^{1/2},\pm \delta) = G_\pm(u) + R_\pm(u,\delta)
$$
with
$
G_\pm(u) = \ii \pm u^2 \alpha
$
and $R_\pm(u,\delta) = \pm \alpha \delta$. Notice that Proposition~\ref{lem:dominant_scale} does not apply here but the situation is simpler since $R_\pm$ do not depend on $u$ and $G_\pm$ do not vanish hence one has $\left|\tfrac{R_\pm}{G_\pm}\right|\to 0$ as $\delta\to 0$, for any $\lambda>0$. Thus we are left with the winding number of 
$$
\mathcal S(u)=\dfrac{G_-(u)}{G_+(u)} = \dfrac{\ii- \alpha  u^2}{\ii+ \alpha  u^2}.
$$
which is 0 as discussed in the case above. Consequently $w_\infty=0$.

\paragraph{Summarizing.} In class $\mathfrak{A}_{1,4}$ one has $w_\infty=\sign(\beta)$ if $\beta \neq 0$ and $w_\infty=0$ if $\beta=0$.

\subsection{Class $\mathfrak{A}_{2,3}$} In this case we have

$$
A_0 = \begin{pmatrix}
	0 & 1 & 0 & 0\\
	\alpha & 0 & 0 & 1
\end{pmatrix}, \qquad A_1= \begin{pmatrix}
	0 & b_{12} & 0 & 0\\
	\ii \beta & b_{22} & 0 & 0
\end{pmatrix} := \begin{pmatrix}
	B & 0 
\end{pmatrix},
$$
with $b_{12}, b_{22} \in \mathbb C$ and $\alpha, \beta \in \mathbb R$. We get $g(k_x,\kappa) = g_\infty(k_x,\kappa)+o(1)$ with
\begin{align}
	g_\infty(k_x,\kappa) =& -\alpha+k_x \left(\beta -\ii \alpha  b_{12}\right)- \det(B) k_x^2 \cr & +\ii \det(B)\frac{k_x^2 \left(k_x-\ii \kappa \right) \left(\kappa _{\text{ev}}-\ii k_x\right)}{4 k^4 \epsilon ^2}-\frac{\ii \left(\kappa _{\text{ev}}-\kappa \right) \left(b_{12} k_x-\ii\right) \left(\kappa _{\text{ev}}-\ii k_x\right)}{2 k^2 \epsilon }
\end{align}
Passing to dual variables via \eqref{eq:defg0} we get
\begin{align}
	g_0(\lambda_x,\delta)= &\, \dfrac{1}{4\epsilon^2(\lambda_x^2+\delta^2)^2} \Big( -4 \alpha  \delta ^4 \epsilon ^2+4 \ii \alpha  b_{12} \delta ^2 \epsilon ^2 \lambda _x+4 \ii \alpha  b_{12} \epsilon ^2 \lambda _x^3-2 \ii b_{12} \delta ^2 \epsilon  \lambda _x\cr &+ 2 b_{12} \delta  \epsilon  \lambda _x \sqrt{\delta ^2+2 \lambda _x^2}-2 \ii b_{12} \epsilon  \lambda _x^2 \sqrt{\delta ^2+2 \lambda _x^2}+2 b_{12} \delta  \epsilon  \lambda _x^2-4 \ii b_{12} \epsilon  \lambda _x^3+2 \delta ^4 \epsilon \cr &+\ii \delta  \det(B) \lambda _x^2 \sqrt{\delta ^2+2 \lambda _x^2}+\det(B) \lambda _x^3 \sqrt{\delta ^2+2 \lambda _x^2}+\ii \delta  \det(B) \lambda _x^3+\det(B) \lambda _x^4-8 \alpha  \delta ^2 \epsilon ^2 \lambda _x^2\cr &-4 \det(B) \epsilon ^2 \lambda _x^2-4 \alpha  \epsilon ^2 \lambda _x^4-4 \beta  \delta ^2 \epsilon ^2 \lambda _x-4 \beta  \epsilon ^2 \lambda _x^3+2 \ii \delta ^3 \epsilon  \sqrt{\delta ^2+2 \lambda _x^2}+2 \ii \delta ^3 \epsilon  \lambda _x+6 \delta ^2 \epsilon  \lambda _x^2\cr &+2 \delta ^2 \epsilon  \lambda _x \sqrt{\delta ^2+2 \lambda _x^2}+2 \ii \delta  \epsilon  \lambda _x^2 \sqrt{\delta ^2+2 \lambda _x^2}+2 \epsilon  \lambda _x^3 \sqrt{\delta ^2+2 \lambda _x^2}+2 \ii \delta  \epsilon  \lambda _x^3+4 \epsilon  \lambda _x^4\Big)
\end{align}
The denominator does not contribute to $w_\infty$ so we replace $g_0$ by $2(\delta ^2+\lambda _x^2)^2g_0$. Then the dominant scale of $g_0$ depends on some parameter vanishing or not. 
\paragraph{Case $\det(B) \neq 0$.} Since $\det(B)=-\ii \beta b_{12}$ then $\beta\neq0$ and $b_{12}\neq 0$. 
 The dominant scale is 
$$ 
\delta^{-4} g_0(u \delta^2,\pm \delta) = G_\pm(u) + R_\pm(u)
$$
with 
\begin{align}
&	G_+(u) = (1+\ii)-2 \alpha  \epsilon + \left(2 \ii \alpha  b_{12} \epsilon -2 \beta  \epsilon +(1-\ii) b_{12}\right)u -2 \epsilon\det(B) u^2,  \cr
& G_-(u) = (1-\ii)-2 \alpha  \epsilon + \left(2 \ii \alpha  b_{12} \epsilon -2 \beta  \epsilon -(1+\ii) b_{12}\right)u -2 \epsilon \det(B) u^2.
\end{align}
The two expressions above factorize to
\begin{align}
	&	G_+(u) = \big((1+\ii) b_{12} u+(-1+\ii)\big) \big((1+i) \alpha  \epsilon +(1+\ii) \beta  u \epsilon -1\big),  \cr
	& G_-(u) = \big((1+\ii) b_{12} u+(-1+\ii)\big) \big((1+i) \alpha  \epsilon +(1+\ii) \beta  u \epsilon -\ii\big).
\end{align}
The first factor is common to $G_+$ and $G_-$ and vanishes at $u_0 = -\ii/b_{12}$. The second factors do not vanish for $u\in \mathbb R$. Thus if $b_{12} \notin \ii\mathbb R$ then Proposition~\ref{lem:dominant_scale} applies and otherwise Proposition~\ref{lem:dominant_scale2} applies, as long as $\det(B)\neq 0$. In both cases the ratio simplifies to
\begin{equation}
	 \dfrac{G_-(u)}{G_+(u)} =  \dfrac{2\epsilon \beta u +2 \alpha  \epsilon  -1+\ii}{2\epsilon \beta  u +2 \alpha  \epsilon  -1-\ii}\,.
\end{equation}
Finally, the map $u \to 2\epsilon \beta  u +2 \alpha  \epsilon  -1+\ii$ winds by $-\sign(\beta)\pi$ as $y$ goes from $-\infty$ to $+\infty$ whereas $u \to 2\epsilon \beta  u +2 \alpha  \epsilon  -1-\ii$ winds by $\sign(\beta)\pi$, so that 
$$w_\infty=-\sign(\beta).$$

\paragraph{Case $\beta=0$ and $b_{12}\neq 0$.} The dominant scale is 
$$ 
\delta^{-3} g_0(u \delta,\pm \delta) = G_\pm(u) + R_\pm(u)
$$
with
\begin{align}
	&G_+(u)= b_{12} u (u+\ii) \left(2 \alpha  \epsilon-1 -\ii (2 u(1-\alpha  \epsilon)  - \sqrt{1+2 u^2})\right)\cr
	&G_-(u) = b_{12} u (u-\ii) \left(2 \alpha  \epsilon-1 +\ii (2 u(1-\alpha  \epsilon) - \sqrt{1+2 u^2})\right) = G^*_+(u)\,.
\end{align}
These two functions vanish at $u=0$, and we can check that Proposition~\ref{lem:dominant_scale2} applies when $b_{12}\neq 0$. Thus we are left with studying the winding phase of 
$$
G(u) = (u-\ii) \left(2 \alpha  \epsilon-1 +\ii (2 u(1-\alpha  \epsilon) - \sqrt{1+2 u^2})\right) := (u-\ii) \widetilde G(u).
$$
The first factor always winds by $\pi$. As for the second factor $\widetilde G(u)$, its real part is constant. Moreover, one has
$$
\widetilde G(u) \mathop{\sim}_{+\infty}  \ii (2(1-\alpha \epsilon)+\sqrt{2} )u, \qquad \widetilde G(u) \mathop{\sim}_{-\infty}  \ii (2(1-\alpha \epsilon)-\sqrt{2} )u\,.
$$
Thus if $\alpha < \tfrac{1}{\epsilon}\left(1-\tfrac{1}{\sqrt 2}\right)$ or $\alpha > \tfrac{1}{\epsilon}\left(1+\tfrac{1}{\sqrt 2}\right)$ one can check that $\widetilde G$ winds by $-\pi$, so that $G$ does not wind. Otherwise, if
$$
\dfrac{1}{\epsilon}\left(1-\dfrac{1}{\sqrt 2}\right)< \alpha < \dfrac{1}{\epsilon}\left(1+\dfrac{1}{\sqrt 2}\right)
$$
then $\widetilde G$ does not wind, so that $G$ winds by $\pi$. To summarize one has
$$
w_\infty = \begin{cases}
	1 & \text{if}\qquad  \frac{1}{\epsilon}\left(1-\frac{1}{\sqrt 2}\right)< \alpha < \frac{1}{\epsilon}\left(1+\frac{1}{\sqrt 2}\right),\\
	0 & \text{if} \qquad  \alpha < \frac{1}{\epsilon}\left(1-\frac{1}{\sqrt 2}\right) \quad \text{or}\quad  \alpha > \frac{1}{\epsilon}\left(1+\frac{1}{\sqrt 2}\right).
\end{cases}
$$

\paragraph{Case $\beta=0$ and $b_{12}=0$.}  The dominant scale is 
$$ 
\delta^{-4} g_0(u \delta,\pm \delta) = G_\pm(u) + R_\pm(u)
$$
with 
\begin{align}
G_-(u)=&1-2 \alpha  \epsilon -2 \alpha  u^4 \epsilon -4 \alpha  u^2 \epsilon +2 u^4+u^3\sqrt{2 u^2+1} +3 u^2+u\sqrt{2 u^2+1} \cr &+\ii \left(-u^3-u^2\sqrt{2 u^2+1} -\sqrt{2 u^2+1}-u\right)
\end{align}
and $G_+(u)= G^*_-(u)$. It is worth noticing here that $$R_+=R-=0$$ so that Proposition~\ref{lem:dominant_scale} is not even required. $G_-(u)$ never vanishes and its imaginary part is always strictly negative. Moreover,
$$
 G(u) \mathop{\sim}_{+\infty}  (2(1-\alpha \epsilon)+\sqrt{2} )u^4, \qquad  G(u) \mathop{\sim}_{-\infty}  (2(1-\alpha \epsilon)-\sqrt{2} )u^4
$$
A similar analysis may be performed as in the previous case, and we end up again with 
$$
w_\infty = \begin{cases}
	1 & \text{if}\qquad  \frac{1}{\epsilon}\left(1-\frac{1}{\sqrt 2}\right)< \alpha < \frac{1}{\epsilon}\left(1+\frac{1}{\sqrt 2}\right),\\
	0 & \text{if} \qquad  \alpha < \frac{1}{\epsilon}\left(1-\frac{1}{\sqrt 2}\right) \quad \text{or}\quad  \alpha > \frac{1}{\epsilon}\left(1+\frac{1}{\sqrt 2}\right).
\end{cases}
$$

\paragraph{Case $b_{12}=0$ and $\beta\neq 0$.} In that case we go back to 
\begin{align}
	g_\infty(k_x,\kappa) = -\alpha+k_x \beta -\frac{ \left(\kappa _{\text{ev}}-\kappa \right)\left(\kappa _{\text{ev}}-\ii k_x\right)}{2 k^2 \epsilon }
\end{align}
Passing to dual variables via \eqref{eq:defg0} we get
\begin{align}
	g_0(\lambda_x,\delta)= &\, \dfrac{-2 \alpha  \delta ^2 \epsilon +\delta ^2-2 \alpha  \epsilon  \lambda _x^2-2 \beta  \epsilon  \lambda _x+\ii \delta  \sqrt{\delta ^2+2 \lambda _x^2}+\lambda _x \sqrt{\delta ^2+2 \lambda _x^2}+\ii \delta  \lambda _x+2 \lambda _x^2}{2\epsilon(\lambda_x^2+\delta^2)}\,.
\end{align}


The denominator does not contribute to $w_\infty$ so we replace $g_0$ by $2(\delta ^2+\lambda _x^2)\epsilon g_0$. The dominant scale is 
$$ 
\delta^{-2} g_0(u \delta^2,\pm \delta) = G_\pm(u) + R_\pm(u)
$$
with 
$$
G_+(u)=-2 \alpha  \epsilon -2 \beta  u \epsilon +1+\ii, \qquad G_-(u)=G^*_+(u).
$$
Proposition~\ref{lem:dominant_scale} applies and we infer that $G_-$ has a winding phase of $-\pi \sign(\beta)$ so that
$$
w_\infty = - \sign(\beta).
$$

\paragraph{Summarizing.} If $\beta \neq 0$ then $w_\infty = - \sign(\beta)$. If $\beta=0$ then 
$$
w_\infty = \begin{cases}
	1 & \text{if}\qquad  \frac{1}{\epsilon}\left(1-\frac{1}{\sqrt 2}\right)< \alpha < \frac{1}{\epsilon}\left(1+\frac{1}{\sqrt 2}\right),\\
	0 & \text{if} \qquad  \alpha < \frac{1}{\epsilon}\left(1-\frac{1}{\sqrt 2}\right) \quad \text{or}\quad  \alpha > \frac{1}{\epsilon}\left(1+\frac{1}{\sqrt 2}\right).
\end{cases}
$$

\begin{rem}
	At the threshold cases $\beta=0$ and $\alpha = \tfrac{1}{\epsilon}\big(1\pm\tfrac{1}{\sqrt 2}\big)$, one can check that the winding phase of $G_+$ is not a multiple of $\pi$ so that $w_\infty \notin \mathbb Z$. We suspect some collapse of edge mode branch occurring at infinity, similarly to the case $a=\pm \sqrt{2}$ in \cite{GrafJudTauber21}. We do not investigate further this fine-tuned case and consider it out of the $w_\infty$-classification. 
\end{rem}

\subsection{Class $\mathfrak{A}_{2,4}$ } In this case we have

$$
A_0 = \begin{pmatrix}
	a_{11} & 1 & 0 & 0\\
	a_{21} & 0 &(a_{11}^*)^{-1}& 1
\end{pmatrix}, \qquad A_1= \begin{pmatrix}
b_{11} & b_{11}(a_{11})^{-1} & 0 & 0\\
b_{21} & b_{22} & 0 & 0
\end{pmatrix} := \begin{pmatrix}
	B & 0 
\end{pmatrix},
$$
with $ a_{11}\in \mathbb C^*$, $b_{11}, b_{21}, b_{22}\in \mathbb C$, such that $a_{21} = \alpha a_{11} + \epsilon^{-1}$ and $ b_{22}-b_{21}a_{11}^{-1}=\ii \beta$ with $\alpha, \beta \in \mathbb R$. We get $g(k_x,\kappa) = g_\infty(k_x,\kappa)+o(1)$ with
\begin{align}
	& g_\infty(k_x,\kappa)= \cr & -\det(B)  k_x^2-\frac{\ii \kappa _{\text{ev}}}{a_{11}^*}+k_x \left(\frac{b_{11} \kappa _{\text{ev}}}{|a_{11}|^2}-\frac{\ii a_{21} b_{11}}{a_{11}}+\ii a_{11} b_{22}+b_{11} (-\kappa )-\ii b_{21}\right)+\ii a_{11} \kappa -a_{21}\cr 
	&-\frac{\left(\kappa _{\text{ev}}-\kappa \right) \left(\kappa _{\text{ev}}-\ii k_x\right)}{2 k^2 \epsilon }+\frac{\left(\kappa _{\text{ev}}-\kappa \right) \left(\kappa +\ii k_x\right) \left(a_{11}+\ii b_{11} k_x\right)}{2 a_{11}^* k^2 \epsilon }-\frac{\ii b_{11} \left(\kappa _{\text{ev}}-\kappa \right) k_x \left(\kappa _{\text{ev}}-\ii k_x\right)}{2 a_{11} k^2 \epsilon }\cr 
	&+\frac{b_{11} k_x \left(k_x-\ii \kappa \right) \left(k_x+\ii \kappa _{\text{ev}}\right) \left(b_{22} k_x+\kappa _{\text{ev}}\right)}{4 k^4 \epsilon ^2}-\frac{b_{11} k_x \left(k_x-\ii \kappa \right) \left(k_x+\ii \kappa _{\text{ev}}\right) \left(b_{21} a_{11}^* k_x+\kappa \right)}{4 |a_{11}|^2 k^4 \epsilon ^2}\,.
\end{align}
Passing to dual variables via \eqref{eq:defg0} we get
\begin{align} \label{eq:longuest_g0}
	& g_0(\lambda_x,\delta)= \cr &\, \dfrac{1}{4 |a_{11}|^2 \epsilon ^2 \left(\delta ^2+\lambda _x^2\right){}^2} \Big( 2 \epsilon  \left| a_{11}\right| {}^2 \delta ^4-2 \epsilon  a_{11}^2 \delta ^4-4 \epsilon ^2 \left| a_{11}\right| {}^2 a_{21} \delta ^4+2 \ii \epsilon  \sqrt{\delta ^2+2 \lambda _x^2} a_{11}^2 \delta ^3\cr &+ 2 \ii \epsilon  \left| a_{11}\right| {}^2 \sqrt{\delta ^2+2 \lambda _x^2} \delta ^3+4 \ii \epsilon ^2 |a_{11}|^2  a_{11} \delta ^3+2 \ii \epsilon  \left| a_{11}\right| {}^2 \lambda _x \delta ^3+2 \ii \epsilon  a_{11}^2 \lambda _x \delta ^3+\left| a_{11}\right| {}^2 b_{11} \lambda _x \delta ^3\cr &+6 \epsilon  \left| a_{11}\right| {}^2 \lambda _x^2 \delta ^2-2 \epsilon  a_{11}^2 \lambda _x^2 \delta ^2-8 \epsilon ^2 \left| a_{11}\right| {}^2 a_{21} \lambda _x^2 \delta ^2-\ii \left| a_{11}\right| {}^2 b_{11} \lambda _x^2 \delta ^2+\ii b_{11} \lambda _x^2 \delta ^2\cr &+4 \epsilon ^2 \sqrt{\delta ^2+2 \lambda _x^2} a_{11} \delta ^2+2 \epsilon  \sqrt{\delta ^2+2 \lambda _x^2} a_{11}^2 \lambda _x \delta ^2+2 \epsilon  \left| a_{11}\right| {}^2 \sqrt{\delta ^2+2 \lambda _x^2} \lambda _x \delta ^2\cr &+\ii \sqrt{\delta ^2+2 \lambda _x^2} b_{11} \lambda _x \delta ^2+2 \ii \epsilon  a_{11} b_{11} \lambda _x \delta ^2+4 \ii \epsilon ^2 \left| a_{11}\right| {}^2 b_{21} \lambda _x \delta ^2-2 \ii \epsilon  b_{11} a_{11}^* \lambda _x \delta ^2\cr &+4 \ii \epsilon ^2 a_{21} b_{11} a_{11}^* \lambda _x \delta ^2-4 \ii \epsilon ^2 a_{11}^2 b_{22} a_{11}^* \lambda _x \delta ^2+2 \ii \epsilon  \left| a_{11}\right| {}^2 \lambda _x^3 \delta +2 \ii \epsilon  a_{11}^2 \lambda _x^3 \delta +2 \left| a_{11}\right| {}^2 b_{11} \lambda _x^3 \delta \cr & +b_{11} \lambda _x^3 \delta +\ii \left| a_{11}\right| {}^2 b_{11} b_{22} \lambda _x^3 \delta -\ii b_{11} b_{21} a_{11}^* \lambda _x^3 \delta +2 \ii \epsilon  \sqrt{\delta ^2+2 \lambda _x^2} a_{11}^2 \lambda _x^2 \delta \cr &+2 \ii \epsilon  \left| a_{11}\right| {}^2 \sqrt{\delta ^2+2 \lambda _x^2} \lambda _x^2 \delta +\left| a_{11}\right| {}^2 \sqrt{\delta ^2+2 \lambda _x^2} b_{11} \lambda _x^2 \delta +\sqrt{\delta ^2+2 \lambda _x^2} b_{11} \lambda _x^2 \delta +2 \epsilon  a_{11} b_{11} \lambda _x^2 \delta \cr &+\ii \left| a_{11}\right| {}^2 \sqrt{\delta ^2+2 \lambda _x^2} b_{11} b_{22} \lambda _x^2 \delta +4 \ii \epsilon ^2 |a_{11}|^2 a_{11} \lambda _x^2 \delta +2 \epsilon  b_{11} a_{11}^* \lambda _x^2 \delta -\ii b_{11} b_{21} a_{11}^* \lambda _x^2 \sqrt{\delta ^2+2 \lambda _x^2} \delta \cr &+4 \epsilon ^2 \left| a_{11}\right| {}^2 b_{11} \lambda _x \delta +2 \epsilon  \sqrt{\delta ^2+2 \lambda _x^2} a_{11} b_{11} \lambda _x \delta +2 \epsilon  \sqrt{\delta ^2+2 \lambda _x^2} b_{11} a_{11}^* \lambda _x \delta +4 \epsilon  \left| a_{11}\right| {}^2 \lambda _x^4\cr &-4 \epsilon ^2 \left| a_{11}\right| {}^2 a_{21} \lambda _x^4-2 \ii \left| a_{11}\right| {}^2 b_{11} \lambda _x^4+\left| a_{11}\right| {}^2 b_{11} b_{22} \lambda _x^4-b_{11} b_{21} a_{11}^* \lambda _x^4+2 \epsilon  \sqrt{\delta ^2+2 \lambda _x^2} a_{11}^2 \lambda _x^3\cr &+2 \epsilon  \left| a_{11}\right| {}^2 \sqrt{\delta ^2+2 \lambda _x^2} \lambda _x^3+4 \ii \epsilon ^2 \left| a_{11}\right| {}^2 b_{21} \lambda _x^3+\left| a_{11}\right| {}^2 \sqrt{\delta ^2+2 \lambda _x^2} b_{11} b_{22} \lambda _x^3-4 \ii \epsilon  b_{11} a_{11}^* \lambda _x^3\cr &+4 \ii \epsilon ^2 a_{21} b_{11} a_{11}^* \lambda _x^3-4 \ii \epsilon ^2 |a_{11}| b_{22}  a_{11} \lambda _x^3-4 \det(B) \epsilon ^2 \left| a_{11}\right| {}^2 \lambda _x^2+4 \epsilon ^2 \sqrt{\delta ^2+2 \lambda _x^2} a_{11} \lambda _x^2\cr& -\ii \left| a_{11}\right| {}^2 b_{11} \lambda _x^3 \sqrt{\delta ^2+2 \lambda _x^2}-b_{11} b_{21} a_{11}^* \lambda _x^3 \sqrt{\delta ^2+2 \lambda _x^2}-2 \ii \epsilon  a_{11} b_{11} \lambda _x^2 \sqrt{\delta ^2+2 \lambda _x^2}\cr &-2 \ii \epsilon  b_{11} a_{11}^* \lambda _x^2 \sqrt{\delta ^2+2 \lambda _x^2}-4 \ii \epsilon ^2 b_{11} \lambda _x \sqrt{\delta ^2+2 \lambda _x^2}\Big).
\end{align}
The denominator does not contribute to $w_\infty$ so we replace $g_0$ by $2\epsilon |a_{11}|^2 (\delta ^2+\lambda _x^2)^2g_0$. Then the dominant scale of $g_0$ depends on some parameter vanishing or not. 

\paragraph{Case $b_{11}\neq0$.} If $b_{11}\neq 0$,
the dominant scale is 
$$ 
\delta^{-2} g_0(u \delta,\pm \delta) = G_\pm(u) + R_\pm(u)
$$
with 
\begin{align}
	&G_+(u) =-2 \det(B) \epsilon |a_{11}|^2 u^2 - 2 \ii \epsilon b_{11} u \sqrt{1+2u^2} + 2 \epsilon b_{11} |a_{11}|^2 u\cr
	&G_-(u) =-2 \det(B) \epsilon |a_{11}|^2 u^2 - 2 \ii \epsilon b_{11} u \sqrt{1+2u^2} - 2 \epsilon b_{11} |a_{11}|^2 u\,.
\end{align}
Recalling that $\det(B) = b_{11} (b_{22}-b_{21}/a_{11})$ and that   $b_{22}-b_{21}/a_{11} = \ii \beta$ with $\beta \in \mathbb R$. $G_+$ and $G_-$ vanish at $u=0$ and we can check that Proposition~\ref{lem:dominant_scale2} applies when $b_{11}\neq 0$. We are left with the ratio 
\begin{equation}
	\mathcal S(u) =  \dfrac{G_-(u)}{G_+(u)} = 
	 \dfrac{\beta |a_{11}|^2 u  + \sqrt{1+2 u^2} - \ii |a_{11}|^2}{\beta |a_{11}|^2 u  + \sqrt{1+2 u^2} + \ii |a_{11}|^2} := \dfrac{G(u)}{G^*(u)}
\end{equation} 
and we are left with twice the winding of the numerator $G(u) = \beta |a_{11}|^2 u  + \sqrt{1+2 u^2} - \ii |a_{11}|^2$ which has constant negative imaginary part so its winding phase is the difference of limits of argument between $+\infty$ and $-\infty$. If $\beta |a_{11}|^2 > \sqrt{2}$ then
$$
\lim_{u \to + \infty}\frac{G(u)}{G(-u)} = -1
$$
and $G$ winds by $\pi$ so that $w_\infty=1$. Similarly, if $\beta |a_{11}|^2 <- \sqrt{2}$ then
 $G$ winds by $-\pi$ so that $w_\infty=-1$. Otherwise, if $|\beta| |a_{11}|^2 < \sqrt{2}$ then
 $$
 \lim_{u \to + \infty}\frac{G(u)}{G(-u)} = 1
 $$
 and $G$ does not wind, so that $w_\infty =0$. 
 
\paragraph{Case $b_{11}=0$.} If $b_{11}=0$, the dominant scale is 
$$ 
\delta^{-3} g_0(u \delta,\pm \delta) = G_\pm(u) + R_\pm(u)
$$
with 
 \begin{align}
& G_+(u) = 2 a_{11} \left(u^2+1\right) \epsilon  \left(\sqrt{2 u^2+1}+\ii \big(|a_{11}|^2+u (a_{11}^*b_{21}-|a_{11}|^2b_{22} )\big)\right)\cr\,, 
& G_-(u) = 2 a_{11} \left(u^2+1\right) \epsilon   \left(-\sqrt{2 u^2+1}+\ii \big((|a_{11}|^2-u (a_{11}^*b_{21}-|a_{11}|^2b_{22} )\big)\right)
 \end{align}
 The self-adjoint condition implies  $b_{22}-b_{21}/a_{11} = \ii \beta$ with $\beta \in \mathbb R$. $G_+$ and $G_-$ do not vanish and Proposition~\ref{lem:dominant_scale} applies.  The ratio  simplifies to:
 $$
\mathcal S(u) =  \dfrac{G_-(u)}{G_+(u)} = \dfrac{-\sqrt{2 u^2+1}+\ii |a_{11}|^2 - \beta |a_{11}|^2 u}{\sqrt{2 u^2+1}+ \ii |a_{11}|^2 +  \beta |a_{11}|^2 u} = -\dfrac{G(u)}{G^*(u)}\,.
 $$
 We are back to the same numerator and denominator from the case $b_{11}=0$, from which we immediately infer  
 $$
 w_\infty = \begin{cases}
 	1, & \beta |a_{11}|^2 > \sqrt{2},\\
 	0,& -\sqrt{2}<\beta |a_{11}|^2 < \sqrt{2}, \\
 	-1,& \beta |a_{11}|^2 <- \sqrt{2}.
 \end{cases}
 $$
 
 \begin{rem}
 	At the threshold cases $\beta|a_{11}|^2 = \pm \sqrt{2}$, one can check that the winding phase of $G_+$ is not a multiple of $\pi$ so that $w_\infty \notin \mathbb Z$. We suspect some collapse of edge mode branch occurring at infinity, similarly to the case $a=\pm \sqrt{2}$ in \cite{GrafJudTauber21}. We do not investigate further this fine-tuned case and consider it out of the $w_\infty$-classification. 
 \end{rem}
 
\subsection{Class $\mathfrak{A}_{3,4}$ \label{sec:anomalyA34}} In this case we have
$$
A_0 = \begin{pmatrix}
	\alpha_{1} & a_{12} & 1 & 0\\
	\einv -a_{12}^* & \alpha_{2} & 0 & 1
\end{pmatrix}, \qquad A_1= \begin{pmatrix}
\ii \beta_{1} & b_{12} & 0 & 0\\
b_{12}^* & \ii \beta_{2} & 0 & 0
\end{pmatrix} := \begin{pmatrix}
	B & 0 
\end{pmatrix},
$$
with $\alpha_{12}, b_{12} \in \mathbb C$ and $\alpha_1, \alpha_2,  \beta_1, \beta_2 \in \mathbb R$. We get $g(k_x,\kappa) = g_\infty(k_x,\kappa)+o(1)$ with
\begin{align}
	g_\infty(k_x,\kappa) = & -\det(B)  k_x^2+\ii \alpha _2 \kappa _{\text{ev}}-\kappa  \kappa _{\text{ev}}+\ii \alpha _1 \kappa +\det(A) \cr 
	&+\ii k_x \left(\ii \alpha _2 \beta _1+\ii \alpha _1 \beta _2-a_{12} b_{12}^*+b_{12} a_{12}^*-\beta _1 \kappa -\frac{b_{12}}{\epsilon }-\beta _2 \kappa _{\text{ev}}\right) \cr 
	&\frac{\left(\kappa _{\text{ev}}-\kappa \right) \left(a_{12} \epsilon  \left(-\kappa _{\text{ev}}+\ii k_x\right)+b_{12} \epsilon  k_x \left(-k_x-\ii \kappa _{\text{ev}}\right)+\left(k_x-\ii \kappa \right) \left(\ii a_{12}^* \epsilon +b_{12}^* \epsilon  k_x-\ii \right)\right)}{2 k^2 \epsilon ^2}\cr 
	&+\frac{\left(k_x-\ii \kappa \right) \left(-\kappa _{\text{ev}}+\ii k_x\right) \left(-\ii \det(B)  k_x^2-\ii \kappa  \kappa _{\text{ev}}+\beta _1 \kappa _{\text{ev}} k_x+\beta _2 \kappa  k_x\right)}{4 k^4 \epsilon ^2}\,.
\end{align}
Passing to dual variables via \eqref{eq:defg0} we get
\begin{align} 
	& g_0(\lambda_x,\delta)= \cr &\, \dfrac{1}{4\epsilon ^2 \left(\delta ^2+\lambda _x^2\right){}^2} \Big( 4 \det(A) \epsilon ^2 \delta ^4+2 \epsilon  a_{12} \delta ^4-2 \epsilon  a_{12}^* \delta ^4+\delta ^4-2 \ii \sqrt{\delta ^2+2 \lambda _x^2} \delta ^3+2 \ii \epsilon  \sqrt{\delta ^2+2 \lambda _x^2} a_{12} \delta ^3\cr &+2 \ii \epsilon  \sqrt{\delta ^2+2 \lambda _x^2} a_{12}^* \delta ^3+4 \ii \epsilon ^2 \alpha _1 \delta ^3-\ii \lambda _x \delta ^3+2 \ii \epsilon  a_{12} \lambda _x \delta ^3+2 \ii \epsilon  a_{12}^* \lambda _x \delta ^3+\ii \beta _1 \lambda _x \delta ^3+8 \det(A) \epsilon ^2 \lambda _x^2 \delta ^2\cr &+6 \epsilon  a_{12} \lambda _x^2 \delta ^2-2 \epsilon  a_{12}^* \lambda _x^2 \delta ^2+\beta _1 \lambda _x^2 \delta ^2+\beta _2 \lambda _x^2 \delta ^2-4 \epsilon ^2 \alpha _2 \sqrt{\delta ^2+2 \lambda _x^2} \delta ^2-3 \lambda _x \sqrt{\delta ^2+2 \lambda _x^2} \delta ^2\cr &+2 \epsilon  \sqrt{\delta ^2+2 \lambda _x^2} a_{12} \lambda _x \delta ^2+2 \ii \epsilon  b_{12} \lambda _x \delta ^2+2 \epsilon  \sqrt{\delta ^2+2 \lambda _x^2} a_{12}^* \lambda _x \delta ^2-4 \ii \epsilon ^2 b_{12} a_{12}^* \lambda _x \delta ^2-2 \ii \epsilon  b_{12}^* \lambda _x \delta ^2\cr &+4 \ii \epsilon ^2 a_{12} b_{12}^* \lambda _x \delta ^2+4 \epsilon ^2 \alpha _2 \beta _1 \lambda _x \delta ^2+\sqrt{\delta ^2+2 \lambda _x^2} \beta _2 \lambda _x \delta ^2+4 \epsilon ^2 \alpha _1 \beta _2 \lambda _x \delta ^2+\det(B) \ii \lambda _x^3 \delta +2 \ii \epsilon  a_{12} \lambda _x^3 \delta \cr &+2 \ii \epsilon  a_{12}^* \lambda _x^3 \delta +2 \ii \beta _1 \lambda _x^3 \delta -\ii \beta _2 \lambda _x^3 \delta +\det(B) \ii \sqrt{\delta ^2+2 \lambda _x^2} \lambda _x^2 \delta +2 \ii \epsilon  \sqrt{\delta ^2+2 \lambda _x^2} a_{12} \lambda _x^2 \delta +2 \epsilon  b_{12} \lambda _x^2 \delta \cr &+2 \ii \epsilon  \sqrt{\delta ^2+2 \lambda _x^2} a_{12}^* \lambda _x^2 \delta -2 \epsilon  b_{12}^* \lambda _x^2 \delta +4 \ii \epsilon ^2 \alpha _1 \lambda _x^2 \delta +\ii \sqrt{\delta ^2+2 \lambda _x^2} \beta _1 \lambda _x^2 \delta -4 \ii \epsilon ^2 \sqrt{\delta ^2+2 \lambda _x^2} \delta \cr &-\ii \lambda _x^2 \sqrt{\delta ^2+2 \lambda _x^2} \delta -\ii \beta _2 \lambda _x^2 \sqrt{\delta ^2+2 \lambda _x^2} \delta -2 \epsilon  b_{12}^* \lambda _x \sqrt{\delta ^2+2 \lambda _x^2} \delta +2 \epsilon  \sqrt{\delta ^2+2 \lambda _x^2} b_{12} \lambda _x \delta +4 \ii \epsilon ^2 \beta _1 \lambda _x \delta \cr &+4 \det(A) \epsilon ^2 \lambda _x^4+\det(B) \lambda _x^4+4 \epsilon  a_{12} \lambda _x^4+2 \beta _1 \lambda _x^4+\det(B) \sqrt{\delta ^2+2 \lambda _x^2} \lambda _x^3+2 \epsilon  \sqrt{\delta ^2+2 \lambda _x^2} a_{12} \lambda _x^3\cr &+2 \epsilon  \sqrt{\delta ^2+2 \lambda _x^2} a_{12}^* \lambda _x^3-4 \ii \epsilon ^2 b_{12} a_{12}^* \lambda _x^3+4 \ii \epsilon ^2 a_{12} b_{12}^* \lambda _x^3+\sqrt{\delta ^2+2 \lambda _x^2} \beta _1 \lambda _x^3+4 \epsilon ^2 \alpha _2 \beta _1 \lambda _x^3+4 \epsilon ^2 \alpha _1 \beta _2 \lambda _x^3\cr &-4 \det(B) \epsilon ^2 \lambda _x^2+2 \ii \epsilon  \sqrt{\delta ^2+2 \lambda _x^2} b_{12}^* \lambda _x^2-2 \lambda _x^3 \sqrt{\delta ^2+2 \lambda _x^2}-2 \ii \epsilon  b_{12} \lambda _x^2 \sqrt{\delta ^2+2 \lambda _x^2}-4 \epsilon ^2 \alpha _2 \lambda _x^2 \sqrt{\delta ^2+2 \lambda _x^2}\cr &-4 \epsilon ^2 \beta _2 \lambda _x \sqrt{\delta ^2+2 \lambda _x^2}\Big).
\end{align}
The denominator does not contribute to $w_\infty$ so we replace $g_0$ by $(\delta ^2+\lambda _x^2)^2g_0$. Then the dominant scale of $g_0$ depends on some parameter vanishing or not. 
\subsubsection{Case $b_{12}\neq 0$}
If $b_{12}\neq 0$ the dominant scale is 
$$ 
\delta^{-2} g_0(u \delta,\pm \delta) = G_\pm(u) + R_\pm(u)
$$
with 
\begin{align}
	&G_+(u) = - \det(B) u^2-\beta _2 u \sqrt{2 u^2+1} -\ii   \sqrt{2 u^2+1}+ \ii  \beta _1 u \cr
	&G_-(u) =- \det(B)   u^2- \beta _2 u \sqrt{2 u^2+1} +\ii   \sqrt{2 u^2+1}- \ii  \beta _1 u 
\end{align}
Notice that $$\det(B) = -\beta_1\beta_2 - |b_{12}|^2 \in \mathbb R.$$ Since $\beta_1 \in \mathbb R$ then $G_+(u)= G^*_-(u)$. Moreover $G_+$ and $G_-$ do not vanish and Proposition~\ref{lem:dominant_scale} applies as long as $\Delta \neq 0$ with 
$$
\Delta :=
 (\det B-\sqrt{2} \beta_2)(\det B+\sqrt{2} \beta_2) =  B_+B_-
$$
with 
$$
B_{\pm} :=\beta_2(\beta_1\pm\sqrt{2}) +|b_{12}|^2.
$$
We are left with studying the argument of 
\begin{equation}
	G(u)= -G_-(u) =  \det(B) u^2 + \beta _2 u \sqrt{2 u^2+1} +  \ii (\beta _1  u-  \sqrt{2 u^2+1}) := G_r(u)+\ii G_i(u).
\end{equation}

\paragraph{Case 1: $\Delta <0$.} In that case one can check that $G'_r(u)$  is nowhere vanishing so that $G_r$ is strictly monotonic. Moreover, the real part dominates the behavior of $G$ asymptotically:
$$
G(u) \mathop{\sim}_{+\infty} (\det(B)+\sqrt{2}\beta_2) u^2, \qquad G(u) \mathop{\sim}_{-\infty} (\det(B)-\sqrt{2}\beta_2) u^2\,.
$$
In particular 
$$
\lim_{u \to \infty} \dfrac{G(u)}{G(-u)} = \dfrac{\Delta}{(\det(B)-\sqrt{2}\beta_2)^2}<0
$$
so that $G$ crosses the whole complex plane near the real line. Consequently, its argument changes by $\pm \pi$.

The exact value of the winding sign of $G$ requires a detailed study of $G_i(u)$, which behaves asymptotically as
$$
G_i(u) \mathop{\sim}_{+\infty} (\beta_1-\sqrt{2})u, \qquad G_i(u) \mathop{\sim}_{-\infty} (\beta_1+\sqrt{2})u.
$$ 
Moreover $G_i(u)$ vanishes only if $\beta_1^2>2$, in which case $G_i(u_0)=0$ with
\begin{equation}\label{eq:A34defu0}
u_0 = \dfrac{\sign(\beta_1)}{\sqrt{\beta_1^2-2}}.
\end{equation}
and
$$
G_r(u_0) = - \dfrac{|b_{12}|^2}{\beta_1^2-2} < 0
$$
since we assume $b_{12}\neq 0$.

\paragraph{Case 1a: $B_+>0$ and $B_-<0$.} 
This implies $B_+-B_-=  2 \sqrt{2} \beta_2 >0$ so that $\beta_2>0$. Moreover, $B_-<0$ implies $\beta_2 (\beta_1-\sqrt2) < -|b_{12}|^2<0$ and therefore $\beta_1-\sqrt2<0$, so that $G_i(u)<0$ near $+\infty$. Moreover, in this case, one has $G_r(u) >0$ near $+\infty$ and $G_r(u)<0$ near $-\infty$. Finally, if $\beta_1 <-\sqrt{2}$ then $G_i(u)>0$ near $-\infty$ and $G_i(u)$ vanishes at $u_0$ given above, with $G_r(u_0)<0$. Putting all together, $G$ winds by $+\pi$. Otherwise, if $\beta_1>-\sqrt{2}$ then $G_i(u)<0$ near $-\infty$ and $G_i(u)$ never vanishes. Again, $G$ winds by $+\pi$.

\paragraph{Case 1b: $B_+<0$ and $B_->0$.} 
Similarly $B_+-B_-<0$ so that $\beta_2 < 0$. Moreover $B_+<0$ implies $\beta_2 (\beta_1+\sqrt2) < -|b_{12}|^2<0$ and therefore $\beta_1+\sqrt 2>0$ so that $G_i(u)<0$ near $-\infty$. Moreover, in this case, one has $G_r(u) <0$ near $+\infty$ and $G_r(u)>0$ near $-\infty$. Finally, if $\beta_1 >\sqrt{2}$ then $G_i(u)>0$ near $+\infty$ and $G_i(u)$ vanishes at $u_0$ given above, with $G_r(u_0)<0$. Putting all together, $G$ winds by $-\pi$. Otherwise, if $\beta_1<\sqrt{2}$ then $G_i(u)<0$ near $+\infty$ and $G_i(u)$ never vanishes. Again, $G$ winds by $-\pi$.

Summarizing cases 1a and 1b we get
$$
\int_{-\infty}^\infty G^{-1}(u)\partial_u G(u)\dd u =  \sign(B_+) \pi
$$
so  that
$$
w_\infty = \sign(B_+) = -\sign(B_-)
$$
 
\paragraph{Case 2: $\Delta >0$.} In that case  one can check that $G'_r(u)$ vanishes exactly once. Moreover, similarly to case 1, we compute
$$
\lim_{u \to \infty} \dfrac{G(u)}{G(-u)} = \dfrac{\Delta}{(\det(B)-\sqrt{2}\beta)^2}>0
$$
This time $G$ starts from $+\infty$ or $-\infty$ and comes back to the same direction. 

\paragraph{Case 2a: $B_+<0$ and $B_-<0$.} These inequalities imply $\beta_2 (\beta_1\pm\sqrt2) < -|b_{12}|^2<0$ which means that $\beta_2\neq 0$ and $\beta_1+\sqrt 2$ and $\beta_1-\sqrt2$ have the same sign. In particular $\beta_1^2-2 >0$ and $G_i(u)$ vanishes at $u_0$ given in \eqref{eq:A34defu0}, with $G_r(u_0)<0$. Moreover, in this case one has $G_r(u) >0$ near $+\infty$ and $-\infty$, so that $G$ fully winds around $0$ as $u$ goes from $-\infty$ to $+\infty$. 

If $\beta_1+\sqrt 2>0$ and $\beta_1-\sqrt2>0$ then $G_i(u)>0$ near $+\infty$ and $G_i(u)<0$ near $-\infty$. Thus $G$ winds by $-2\pi$ around zero, so that $w_\infty=-2$. Otherwise, if $\beta_1+\sqrt 2<0$ and $\beta_1-\sqrt2<0$ then $G_i(u)<0$ near $+\infty$ and $G_i(u)>0$ near $-\infty$. Thus $G$ winds by $2\pi$, so that $w_\infty=+2$. Summarizing,
$$
w_\infty = -2 \sign(\beta_1+\sqrt 2) = -2 \sign(\beta_1-\sqrt 2).
$$

\paragraph{Case 2b: $B_+>0$ and $B_->0$.}   In this case one has $G_r(u) <0$ near $+\infty$ and $-\infty$. Either $G_i(u)$ vanishes once at $u_0$, with $G_r(u_0)<0$, or $G_i(u)$ never vanishes. In both cases, $G$ never crosses the real positive axis, so that its change of argument is zero. Consequently
$$
w_\infty =0.
$$
Notice that this case is also valid when $\beta_2=0$, or $\beta_1 = \pm \sqrt{2}$.

\paragraph{Summarizing}

If $b_{12}\neq0$ and $\Delta \neq 0$ then we have the following table:
\begin{equation}
	\begin{array}{|c|c|c|}
		\hline
		B_+ & B_- & w_\infty \\
		\hline
		>0 & <0 & 1\\ \hline
		<0 & >0 & -1\\ \hline
		>0 & >0 & 0 \\ \hline 
		<0 & <0 & \pm 2\\ \hline
	\end{array}
\end{equation}
with $\pm 2 = 2 \sign(\sqrt 2 -\beta_1)$.

\subsubsection{Case $b_{12}=0$, $\beta_2\neq 0$ and $\beta_1^2 \neq 2$.} In that case the dominant scale, $G_\pm$ and $R_\pm$ are the same as before, except that $G$ further simplifies to
$$
G(u)=  (\beta _1  u-  \sqrt{2 u^2+1})(-\beta_2 u + \ii).
$$
The first factor vanishes at $u_0 = \sign{\beta_1} (\beta_1^2-2)^{-1/2}$, only if $\beta_1^2 > 2$, in which case Proposition~\ref{lem:dominant_scale2} applies. Otherwise Proposition~\ref{lem:dominant_scale} applies as long as $\beta_1^2 < 2$. In both cases we are left with
$$
\mathcal S(u)= \dfrac{G(u)}{G^* (u)} = \dfrac{-\beta_2 u + \ii}{-\beta_2 u - \ii}
$$
so that, as in previous classes, we infer
$$
w_\infty = \sign(\beta_2).
$$

\subsubsection{Case $b_{12}=0$, $\beta_2= 0$ and $\beta_1^2 \neq 2$.} Notice that in that case one has $\Delta=0$. The dominant scale is 
$$ 
\delta^{-3/2} g_0(u \delta^{1/2},\pm \delta) = G_\pm(u) + R_\pm(u)
$$
with 
$$
G_+(u)=  \left(\alpha _2 u^2+\ii\right) \left(\beta _1 u-\sqrt{2} \sqrt{u^2}\right), \qquad G_-(u)=G^*(u).
$$
The function $G_+$ and $G_-$ vanish at $u_0=0$. One can check that Proposition~\ref{lem:dominant_scale2} applies as long as $\beta_1^2\neq 2$, in which case we are left with the winding of
$$
\mathcal S(u) = \dfrac{G_-(u)}{G_+(u)} = \dfrac{\alpha _2 u^2-\ii}{\alpha _2 u^2+\ii},
$$
from which we infer $w_\infty=0$.

\paragraph{Remaining cases.} We did not deal with the cases where $\Delta =0$ with
$$
\Delta = \big((\beta_1+\sqrt{2}) \beta_2 + |b_{12}|^2\big) \big((\beta_1-\sqrt 2)\beta_2 + |b_{12}|^2\big)
$$
except if $\beta_2=b_{12}=0$, which is treated in the last case above. The four remaining cases would be
\begin{enumerate}
	\item $\beta_2\neq 0$, $b_{12}=0$ and $\beta_1=  \pm \sqrt{2}$,
	\item $\beta_1^2\neq 2$, $b_{12}\neq 0$ and
	$$
	\beta_2= -\dfrac{|b_{12}|^2}{\beta_1\pm\sqrt{2}}.
	$$ 
\end{enumerate}
The same issue with all this cases is that the dominant scale is correct at $u \to +\infty$ and fails at $u\to -\infty$ (or conversely), so that neither Proposition~\ref{lem:dominant_scale} or \ref{lem:dominant_scale2} apply. One could still compute $w_\infty$ by splitting the integral into two parts (positive and negative $\lambda_x$), study them separately with distinct dominant scales, then compute ``half-winding'' phases and finally glue them together properly, but we do not expect $w_\infty \in \mathbb Z$. We rather suspect some collapse of edge mode branch occurring at infinity, similarly to the case $a=\pm \sqrt{2}$ in \cite{GrafJudTauber21}. We do not investigate further this fine-tuned case and consider it out of the $w_\infty$-classification.

\subsection{Class $\mathfrak{B}$} In this case we have

$$
A_0 = \begin{pmatrix}
	a_1 & a_2 & \ii\alpha & - \ii  \mu^*\alpha\\
	\mu a_1 & \mu a_2 & \ii  \mu  \alpha & - \ii |\mu|^2  \alpha
\end{pmatrix}, \qquad A_1= \begin{pmatrix}
1 & 0 & 0 & 0\\
0 & 1 & 0 & 0
\end{pmatrix},
$$
with $a_{1}, a_{2}, \mu \in \mathbb C$ and $\alpha \in \mathbb R$. The self-adjoint condition further requires $\alpha=0$ or $\alpha \Im(\mu) - \epsilon \Re(a_1-a_2\mu)=0$ but we shall keep it as an implicit constraint and keep general $a_2,\, \alpha$ for the computations. We get $g(k_x,\kappa) = g_\infty(k_x,\kappa)+o(1)$ with
\begin{align}
	g_\infty(k_x,\kappa) = &-k_x^2 +\ii a_2 \mu  k_x+\ii a_1 k_x+\frac{\ii \kappa _{\text{ev}} k_x^3+\kappa  \kappa _{\text{ev}} k_x^2-\ii \kappa  k_x^3+k_x^4}{4 k^4 \epsilon ^2}\,.
\end{align}
Passing to dual variables via \eqref{eq:defg0} we get
\begin{align} 
	 g_0(\lambda_x,\delta)=  &\, \dfrac{1}{4\epsilon ^2 \left(\delta ^2+\lambda _x^2\right){}^2} \Big( -4 \ii a_2 \delta ^2 \mu  \epsilon ^2 \lambda _x-4 \ii a_1 \delta ^2 \epsilon ^2 \lambda _x-4 \ii a_2 \mu  \epsilon ^2 \lambda _x^3-4 \ii a_1 \epsilon ^2 \lambda _x^3+\lambda _x^3 \sqrt{\delta ^2+2 \lambda _x^2}\cr &+\ii \delta  \lambda _x^2 \sqrt{\delta ^2+2 \lambda _x^2}+\ii \delta  \lambda _x^3+\lambda _x^4-4 \epsilon ^2 \lambda _x^2\Big).
\end{align}
The denominator, as well as a common $\lambda_x$ factor, do not contribute to $w_\infty$ so we replace $g_0$ by $(\delta ^2+\lambda _x^2)^2 \lambda_x^{-1} g_0$.  The dominant scale is 
$$ 
\delta^{-2} g_0(u \delta^2,\pm \delta) = G_\pm(u) + R_\pm(u)
$$
with 
\begin{equation}
	G_+(u) = -\ii a_2 \mu -\ii a_1-u, \qquad  G_-(u) = G_+(u).
\end{equation}
$G_+$ and $G_-$ do not vanish if $a_2 \mu +a_1 \notin \ii \mathbb R$, and we can check that Proposition~\ref{lem:dominant_scale} applies, from which we immediately infer $\mathcal S=1$ and $w_\infty=0$. If $a_2 \mu +a_1 \in \ii \mathbb R$, since $G_+=G_-$ then we can replace $g_0$ by $g_0+\ii \eta$ with $\eta\in \mathbb R^*$ sufficiently small so that the winding phases of $g_0(\lambda_x,\delta)$ and $g_0(\lambda_x,-\delta)$ are unchanged. Then Proposition~\ref{lem:dominant_scale} applies to $G_+ + \ii \eta = G_-+\ii \eta$. Again, we get $w_\infty=0$.

\subsection{Class $\mathfrak{C}$} In this case we have

$$
A_0 = \begin{pmatrix}
	a_1 & a_2 & 0 & a_4\\
	\mu a_1 & \mu a_2 & 0 & \mu a_4
\end{pmatrix}, \qquad A_1= \begin{pmatrix}
	1 & 0 & 0 & 0\\
	0 & 0 & 0 & 0
\end{pmatrix},
$$
with $a_{1}\in \mathbb C$,  $a_{2}, a_4 \in \mathbb C^2\setminus\{0\}$ and $\mu \in \mathbb C^*\setminus\{0\}$. The self-adjoint condition further requires $\Im(a_2 a_4^*)=0$, which implies $a_2=0$, or $a_4=0$ or $a_4=r a_2$ with $r \in \mathbb R^*$ and $a_2 \neq 0$. We get $g(k_x,\kappa) = g_\infty(k_x,\kappa)+o(1)$ with
\begin{align}
	g_\infty(k_x,\kappa) = &-a_4  \mu  k_x \kappa+\ii a_2 \mu  k_x+\frac{a_4 \mu  k_x \left(\kappa  \kappa _{\text{ev}}^2+\kappa _{\text{ev}} k_x^2+\ii \kappa _{\text{ev}}^2 k_x-\ii \kappa  \kappa _{\text{ev}} k_x\right)}{4 k^4 \epsilon ^2}\,.
\end{align}
Passing to dual variables via \eqref{eq:defg0} we get
\begin{align}
	g_0(\lambda_x,\delta)= &\, \dfrac{1}{4\epsilon^2(\lambda_x^2+\delta^2)^2} \Big( -\ii a_4 \mu  \lambda _x^3 \sqrt{\delta ^2+2 \lambda _x^2}-\ii a_4 \delta ^2 \mu  \lambda _x^2+a_4 \delta  \mu  \lambda _x^2 \sqrt{\delta ^2+2 \lambda _x^2}+a_4 \delta ^3 \mu  \lambda _x\cr &-4 \ii a_2 \delta ^2 \mu  \epsilon ^2 \lambda _x+2 a_4 \delta  \mu  \lambda _x^3+4 a_4 \delta  \mu  \epsilon ^2 \lambda _x-2 \ii a_4 \mu  \lambda _x^4-4 \ii a_2 \mu  \epsilon ^2 \lambda _x^3\Big)\,.
\end{align}
The denominator, as well as some common $\lambda_x \mu$ factor, do not contribute to $w_\infty$ so we replace $g_0$ by $\lambda_x^{-1}\mu^{-1}(\delta ^2+\lambda _x^2)^2g_0$. Then the dominant scale of $g_0$ depends on some parameter vanishing or not. 
\paragraph{Case $a_2 \neq 0$ and $a_4=r a_2$ with $r \in \mathbb R^*$.} The dominant scale is 
$$ 
\delta^{-1} g_0(u \delta^{1/2},\pm \delta) = G_\pm(u) + R_\pm(u)
$$
with 
\begin{align}
	& G_+(u) = a_2(1-\ii r u^2),\cr 
	& G_-(u) = a_2(1+\ii r u^2) = G^*_+(u).
\end{align}
Proposition~\ref{lem:dominant_scale} applies, and $G_-(u)$ has no winding phase as $u$ goes from $-\infty$ to $+\infty$, so that $w_\infty=0$.

\paragraph{Case $a_2=0$ and $a_4 \neq 0$.}
In that case we also divide $g_0$ by $a_4$. The dominant scale is 
$$ 
\delta^{-1} g_0(u \delta^{1/3},\pm \delta) = G_\pm(u) + R_\pm(u)
$$
with 
$G_+(u)=1 - \ii \dfrac{2 u + \sqrt2 |u|}{4 \epsilon^2} u^2, \qquad 
G_-(u) =1 + \ii \dfrac{2 u + \sqrt2 |u|}{4 \epsilon^2} u^2 = G^*_+(u)$. 
Proposition~\ref{lem:dominant_scale} does not apply but $G_+$ and $G_-$ do not vanish and 
$$
\dfrac{R+(u,\delta)}{G_+(u)} \mathop{\sim}_{u \to \pm \infty} \dfrac{\ii \delta^{2/3}}{u}, \qquad \dfrac{R-(u,\delta)}{G_-+(u)} \mathop{\sim}_{u \to \pm \infty} \dfrac{-\ii \delta^{2/3}}{u}
$$
so that  $\left|\tfrac{R_\pm}{G_\pm}\right|\to 0$ as $\delta \to 0$ and all $u \in [-\tfrac{\lambda}{\delta^{1/3}},\tfrac{\lambda}{\delta^{1/3}}]$, and we get the same conclusion than Proposition~\ref{lem:dominant_scale}. The real part of $G_-$ is constant and positive, and  
		$$
		G_-(u)  \mathop{\sim}_{u \to \pm \infty} \dfrac{2\pm \sqrt 2}{4\epsilon^2}u^3
		$$
so that the winding phase of $G_-$ is $\pi$, and since $G_+=G^*_-$ we infer $w_\infty =1$ for any $a_4 \neq 0$.

\paragraph{Case $a_2\neq 0$ and $a_4=0$.}

The dominant scale is 
$$ 
\delta^{-2} g_0(u \delta,\pm \delta) = G_\pm(u) + R_\pm(u)
$$
with 
$$
G_+(u)=-\ii a_2 \left(u^2+1\right), \qquad G_-(u)=G_+(u)
$$
and $R_\pm =0$ here. Thus we immediately infer $w_\infty=0$.

\paragraph{Summarizing.} In that class, $w_\infty =1$ if $a_2 =0$ and $a_4\neq 0$, and $w_\infty=0$ otherwise.

\appendix

\section{Chern number \label{app:Chern}} 

We rewrite the Hamiltonian in terms of Pauli matrices as
\begin{equation}
	H = \vec{d}\cdot \vec{\sigma}\,,\quad \vec{d}=\begin{pmatrix}
		-k_x\\
		-k_y\\
		m-\epsilon \bvec{k}^2
	\end{pmatrix}\,,
\end{equation}
where $\vec{\sigma}=(\sigma_1,\sigma_2,\sigma_3)$ is a vector of Pauli matrices
\begin{equation}
	\sigma_1 =\begin{pmatrix}
		0 & 1\\
		1 & 0 
	\end{pmatrix}\,,\quad \sigma_2 = \begin{pmatrix}
		0 & -\ii\\
		\ii & 0
	\end{pmatrix}\,,\quad \sigma_3 =\begin{pmatrix}
		1 & 0\\
		0 & -1
	\end{pmatrix}\,.
\end{equation}
The eigenprojections of $H$ are shared with those of the flat Hamiltonian $H'=\vec{e}\cdot\vec{\sigma}$, where $\vec{e}=\vec{d}/\abs{\vec{d}}$. They are
\begin{equation}
	P_\pm =\frac{1}{2}\left(1\pm\vec{e}\cdot\vec{\sigma}\right)\,,
\end{equation}
($\left(\vec{a}\cdot\vec{\sigma}\right)\left(\vec{b}\cdot\vec{\sigma}\right)=\left(\vec{a}\cdot\vec{b}\right)+\ii\left(\vec{a}\times\vec{b}\right)\cdot\vec{\sigma}$). Note that $\vec{e}=\vec{e}(\bvec{k})$ is convergent for $k\to\infty$
\begin{equation}
	\vec{e}\to \begin{pmatrix}
		0\\
		0\\
		\pm 1
	\end{pmatrix}\quad (k\to\infty)\,, \quad \mathrm{for}\,\, \epsilon\gtrless 0\,.
\end{equation}
Therefore, also the eigenprojections converge and the Chern number is a well-defined topological invariant
\begin{equation}
	C(P)=\frac{1}{2\pi\ii}\int_{\mathbb{R}^2} dk_xdk_y\tr\left(P\left[\partial_{k_x}P,\partial_{k_y}P\right]\right)\,.
\end{equation}
If the regulator $\epsilon\neq 0$ we can compactify the momentum plane to the 2-sphere $S^2$ and we can compute the r.h.s. on a closed manifold with the map $\vec e : S^2 \to S^2$. According to Prop.~1 of~\cite{GrafJudTauber21}, we get 
$$
C_\pm= \pm \dfrac{1}{4\pi} \int_{S^2}  \vec e \cdot ( \partial_1 \vec e \wedge \partial_2 \vec e)\, \dd x_1 \dd x_2\,,
$$
which, in our case, leads to
\begin{equation}
	C_\pm = \pm \frac{\sign{m}+\sign{\epsilon}}{2}\,.
\end{equation}

\end{document}